\documentclass[11pt, onecolumn]{IEEEtran}
\usepackage{setspace, verbatim, amsfonts, graphicx, amsmath, amsbsy, amssymb, epsfig, url}
\usepackage{hhline}
\usepackage[normalem]{ulem}
\usepackage{booktabs}
\usepackage{multirow, array}
\usepackage{epstopdf}
\usepackage{algorithm,algpseudocode}

\usepackage{cite}
\usepackage{makecell}
\usepackage{epstopdf}

\usepackage{amsfonts}
\usepackage{amsmath,bm}
\usepackage{amssymb}
\usepackage{amsthm}
\usepackage{tikz,graphicx}
\usepackage{mathrsfs} 
\usepackage{algpseudocode}
\usepackage{algorithm}
\usepackage{array,multirow}
\usepackage{color}
\usepackage{slashbox} 			
\usepackage{diagbox}
\usepackage{caption}
\usepackage{subcaption}

\newcommand{\ab}{{\mathbf a}}

\newcommand{\fb}{{\mathbf f}}
\newcommand{\gb}{{\mathbf g}}
\newcommand{\hb}{{\mathbf h}}

\newcommand{\lb}{{\mathbf l}}
\newcommand{\mb}{{\mathbf m}}

\newcommand{\vb}{{\mathbf v}}

\newcommand{\xb}{{\mathbf x}}
\newcommand{\yb}{{\mathbf y}}

\newcommand{\Yc}{\mathcal{Y}}

\newcommand{\Sc}{\mathcal{S}}

\newcommand{\Cd}{{\mathbb C}}

\newcommand{\Zd}{\mathbb{Z}}
\newcommand{\rank}{\textsc{rank}}
\newcommand{\Null}{\textsc{Nul}}
\newcommand{\hank}{\mathscr{H}}
\newcommand{\conv}{\mathscr{C}}

\newcommand{\zerob}{\mathbf{0}}
\newcommand{\hbk}{\boldsymbol{\mathfrak{h}}}

\newcommand{\vect}{\textsc{Vec}}

\theoremstyle{definition} \theoremstyle{plain}

\def\x{{\mathbf x}}

\newtheorem{theorem}{Theorem}[section]
\newtheorem{lemma}[theorem]{Lemma}
\newtheorem{proposition}[theorem]{Proposition}

\usepackage{xcolor}
\newcommand\ruline{\bgroup\markoverwith{\textcolor{red}{\rule[-0.7ex]{2pt}{0.8pt}}}\ULon}

\begin{document}


\title{A general framework for compressed sensing and parallel MRI using annihilating filter based low-rank Hankel matrix
\author{Kyong Hwan Jin, Dongwook Lee, and Jong Chul Ye,~\IEEEmembership{Senior Member,~IEEE}}
\thanks{The authors are with Dept. of Bio and Brain Engineering,  KAIST, Daejeon 305-701, Republic of Korea.  Email:jong.ye@kaist.ac.kr.}
\thanks{This work was supported by Korea Science and Engineering Foundation under Grant NRF-2014R1A2A1A11052491.}}



%
\maketitle
%
\begin{abstract}
\setstretch{1}
Parallel MRI (pMRI) and compressed sensing MRI (CS-MRI) have been considered as two distinct reconstruction problems. 
Inspired by recent k-space interpolation methods, an annihilating filter based low-rank Hankel matrix approach (ALOHA)  is proposed as a general framework for sparsity-driven k-space interpolation method which unifies pMRI and CS-MRI.
Specifically, our framework  is based on the fundamental duality between the transform domain sparsity in the primary space and the low-rankness of {\em weighted} Hankel matrix in the reciprocal space,
which converts pMRI and CS-MRI to a k-space interpolation problem using structured matrix completion.
%
Using  theoretical results from the latest compressed sensing literatures,  we showed that the required sampling rates
for ALOHA may achieve the optimal rate. 
 Experimental results with in vivo data for single/multi-coil imaging as well as dynamic imaging confirmed that the proposed method outperforms the state-of-the-art pMRI and CS-MRI.
\end{abstract}



\vspace*{1cm}

\noindent{
Correspondence to: \vspace*{0.2cm}\\ 
 Jong Chul Ye,  Ph.D. ~~\\
Professor \\
\vspace{-0.0cm}Dept. of Bio and Brain Engineering,  KAIST \\
\vspace{-0.0cm}291 Daehak-ro Yuseong-gu, Daejon 305-701, Republic of Korea \\
\vspace{-0.0cm}Email: jong.ye@kaist.ac.kr \\
\vspace{-0.0cm}Tel: 82-42-350-4320 \\
\vspace{-0.0cm}Fax: 82-42-350-4310 \\
 }

\vspace{\stretch{1}}


\newpage
\baselineskip 0.29in
\section{Introduction}

Magnetic resonance imaging (MRI) is an imaging system that sequentially acquires  k-space data corresponding to the Fourier transform of an object.
This
 enables us to apply various advanced signal processing techniques. Recently,  compressed sensing theory \cite{donoho2006compressed,candes2006robust} has been used in accelerated MRI \cite{lustig2007sparse,liang2009accelerating,haldar2011compressed}.  Compressed sensing algorithms can restore original signals from much less k-space data by exploiting the sparsity of an unknown image in  total variation (TV) or wavelet transform domains, and incoherent sampling schemes such as Gaussian random or Poisson disc are usually required.
Accurate MRI reconstruction from less data  makes compressed sensing a hot topic in the research community; thus, it has been applied across many  different application areas such as in pediatric imaging \cite{vasanawala2010improved}, dynamic cardiac MRI\cite{jung2009k,jung2010radial,yoon2014motion}, perfusion imaging\cite{lingala2011accelerated}, angiography\cite{ccukur2009improving}, and so on.

On the other hand, parallel MRI (pMRI) \cite{pruessmann1999sense,griswold2002generalized,pruessmann2006encoding} exploits the diversity  in the  receiver coil sensitivity maps
that are multiplied by an unknown image.
This provides additional spatial information for the unknown image, resulting in  accelerated MR data acquisition through  k-space sample reduction.
Representative parallel imaging algorithms such as SENSE (sensitivity encoding) \cite{pruessmann1999sense} or GRAPPA (generalized autocalibrating partially parallel acquisitions) \cite{griswold2002generalized} require regularly sampled k-space data for computationally efficient reconstruction.
Moreover, additional k-space data, the so-called  auto calibration  (ACS) lines, are often required to estimate the coil sensitive maps or GRAPPA kernels \cite{griswold2002generalized}.

Because the aim of the two approaches is accelerated acquisition by reducing the k-space data, extensive research efforts have been made to synergistically combine the two for further acceleration. One of the most simplest approaches  can be a SENSE type approach that explicitly utilizes the estimated coil maps to obtain an augmented compressed sensing problem:
\begin{eqnarray}\label{eq:sense}
\min_{\fb} \|W\fb\|_1 && \mbox{subject to} \quad  \gb =  \begin{bmatrix}\gb_1 \\ \vdots \\\gb_r\end{bmatrix} = \begin{bmatrix} A [S_1] \\ \vdots \\ A [S_r] \end{bmatrix} \fb 
\end{eqnarray}
where $\fb$ and $\gb_i$ denote the unknown image and the k-space measurements from the $i$-th coil, respectively; $A$ is a  subsampled Fourier matrix; 
$W$ is a sparsifying transform, and $[S_i]$ denotes a diagonal matrix
whose diagonal elements come from the $i$-th coil sensitivity map.
The multichannel version of k-t FOCUSS \cite{jung2009k} is one of the typical examples of such approaches.  
On the other hand, $l_1$-SPIRiT ($l_1$- iTerative Self-consistent Parallel Imaging Reconstruction)  \cite{lustig2010spirit} utilizes the GRAPPA type constraint as an additional constraint  for a compressed sensing problem:
\begin{eqnarray}
\min_{F} &&\|\Psi F \|_{1,2} \\
\mbox{subject to}&&\quad G= A F\\
&& \vect(F)=M \cdot \textsc{Vec}(F)
\end{eqnarray}
where $\|\cdot\|_{1,2}$ denotes the $(1,2)$-mixed norm of a matrix, 
$F=\begin{bmatrix}\fb_1 & \fb_2 & \cdots &\fb_r\end{bmatrix}, ~G=\begin{bmatrix}\gb_1 & \gb_2 & \cdots & \gb_r\end{bmatrix}$ denote the
unknown images and their k-space measurements for the given set of coils,
and $\Psi$ denote a discrete wavelet transform matrix, and $M$ is an image domain GRAPPA operator, and $\vect(\cdot)$ is the vectorization operator.
In both approaches, an accurate estimation of coil sensitivity maps or GRAPPA kernel is essential to fully exploit the coil sensitivity diversity.   

In order to overcome these difficulties,  calibration-less parallel imaging methods have been extensively investigated, among which SAKE (simultaneous autocalibrating and k-space estimation)\cite{shin2013calibrationless} represents one of the first steps. In SAKE, 
the missing k-space elements are reconstructed by imposing the data consistency and the structural maintenance constraints of the block Hankel structure matrix. However, the origin of the low rankness in the Hankel structured matrix for the case of a single coil measurement was not extensively investigated, and it was not clear whether SAKE could outperform the compressed sensing approach when it is applied to single coil data.
Haldar \cite{haldar2014low,haldar2014ploraks} later discovered  that a Hankel structured matrix  constructed by a single coil k-space measurement is low-ranked  when an unknown image has finite support or a slow-varying phase. Based on this observation, he developed the so called  LORAKS (Low-rank modeling of local k-space neighborhoods) algorithm \cite{haldar2014low} and its parallel imaging version, P-LORAKS (Low-rank modelling of local k-space neighborhoods with parallel imaging data) \cite{haldar2014ploraks }.
However,  it is not clear how the existing  theory can  deal with large classes of image models that  are not finite supported but  can be  sparsified using various transforms such as wavelet transforms or total variations (TV), etc.

Therefore, one of the main goals of this paper is to  develop a  theory that
unifies and generalizes k-space low-rank methods to also allow for transform sparsity models which are critical for practical MRI applications. 
Toward this goal,  we show that  the transform domain sparsity in the signal space can be directly related to the existence of {\em annihilating filters}  \cite{vetterli2002sampling,dragotti2007sampling,maravic2005sampling} in the {\em weighted} k-space.
Interestingly, the commutative relation between an annihilating filter and weighted k-space measurements provides a rank-deficient Hankel structured matrix, whose rank is
determined by the sparsity level of the underlying signal in  a transform domain.
Therefore, by performing a low-rank matrix completion approach, the missing weighted k-space data in the Hankel structured matrix can be recovered, after which the
original k-space data can be recovered by removing the weights.

Interestingly, the new framework is so general that it can cover important  compressed sensing MR approaches in very unique ways.
For example, 
 when an image can be sparsified with wavelet transforms,  
 the low rank structured matrix completion problem  can be solved using a pyramidal decomposition after applying 
scale dependent  k-space weightings. 
Specifically,   a low rank  Hankel matrix completion algorithm can be progressively applied from fine to coarse scale  to reduce the overall computational burden while
maintaining noise robustness.
In addition,  
we show that there exist additional inter-coil annihilating filter relationships that are unique in pMRI,  which can be
 utilized to construct a concatenated Hankel matrix which is low ranked.
We further substantiate that the
multi-channel stacking of the weighted Hankel structure matrix may fully exploit the coil diversity thanks to the relationship to the algebraic bound of multiple measurement
vector (MMV) compressed sensing \cite{kim2012compressive,lee2012subspace,davies2012rank}.
%

Another important advantage of the proposed algorithm is that,  compared to the existing CS-MRI, the reconstruction errors are  usually scattered throughout the entire images rather than exhibiting systematic distortion along edges because the annihilating filter relationships are specifically designed for estimating the edge signals.
 Given that many diagnostic errors are caused by the systematic distortion of images,
 we believe that our ALOHA framework may have great potential in clinical applications.

The remainder of this paper is as follows. Section~\ref{sec:theory} discusses the  fundamental duality between the transform domain sparsity and the low-rankness of Hankel structured
matrix in k-space.
In Section \ref{sec:mri}, a detailed description of the proposed method for MR reconstruction will be provided. 
Section~\ref{sec:method} then explains the implementation detail.
Experimental results are provided in Section \ref{sec:result}, which is followed by the discussion in Section~\ref{sec:discussion} and conclusion in Section \ref{sec:conclusion}.

%


\section{Fundamental Duality between Sparsity and Low-Rankness} 
\label{sec:theory}

This section describes the  fundamental dual relationship between transform domain sparsity and low rankness of {\em weighted} Hankel matrix,  which is the key idea of the proposed algorithm.
For better readability,  the theory here is  outlined by assuming 1-D signals, 
but  the principle can be extended for multidimensional signals.

Note that typical signals of our interest may not be sparse in the image domain, but can be sparsified in a transform domain. 
For example, consider 
$\mathrm{L}$-spline signal model 
from the theory of {\em sparse stochastic processes}  \cite{unser2014unified,unser2014unified2}. Specifically, the signal $f$ of our interest is assumed to satisfy the 
following  partial differential equation: 
\begin{equation}\label{eq:ssp}
\mathrm{L} f = w
\end{equation}
where $\mathrm{L}$ denotes a constant coefficient linear  differential equation (or whitening operator in \cite{unser2014unified,unser2014unified2}): 
\begin{equation}
\mathrm L := a_K\mathrm \partial^K+a_{K-1}\mathrm \partial^{K-1}+\ldots+a_1\mathrm \partial+a_0
\end{equation}
and $w$ is a driving continuous domain sparse signal given by 
\begin{eqnarray}\label{eq:w}
w(x) = \sum_{j=0}^{k-1} c_j \delta \left( x-x_j \right) \, \quad x_j \in [0, \tau].
\end{eqnarray}
Here, without loss of generality, we set $\tau=n_1$ for a positive integer  $n_1 \in \Zd$.
This model includes many class of signals with finite rate of innovations \cite{vetterli2002sampling,dragotti2007sampling,maravic2005sampling}.
For example, if the underlying signal is piecewise constant, we can set $\mathrm L$ as the first differentiation. In this case, $f$ corresponds to the total variation (TV) signal model, and this TV signal model
will be extensively used throughout the paper.

Now, by taking the Fourier transform of \eqref{eq:ssp}, we have
\begin{eqnarray}\label{eq:y}
\hat y(\omega) := {\cal F}\{ \mathrm L f(x) \} = \hat l (\omega) \hat f(\omega)  = \sum_{j=0}^{k-1} a_j e^{-i\omega x_j}
\end{eqnarray}
where 
\begin{eqnarray}
\hat l(\omega) =  a_K (i\omega)^K +a_{K-1} (i\omega)^{K-1}+\ldots+a_1(i\omega)+a_0
\end{eqnarray}
In standard Nyquist sampling, 
we  should measure discrete set of  Fourier samples
at $\{\omega_i\}_{i=1}^m$ from a deterministic  grid, whose
  grid size should be set to  
the Nyquist limit
$\Delta = 2\pi/n_1$ to avoid aliasing artifacts;
 so
 the discrete specturm can be represented as
\begin{eqnarray}\label{eq:spec}
\hat y[m] := \left.\hat y(\omega)\right|_{\omega = m\Delta}  = \hat l[m\Delta]\hat f[m\Delta]
 =   \sum_{j=0}^{k-1} c_j e^{-i2\pi m x_j/n_1 } , \quad 
\end{eqnarray} 
for $m \in [0, \cdots, n_1-1] $.
The discrete spectral sampling model in Eq.~\eqref{eq:spec} implies that 
 the unknown signal in the image domain is  an infinite periodic streams of Diracs with a period $n_1$, 
which
is indeed a signal with the finite rate of innovation  (FRI) with rate $\rho = 2k/n_1$ \cite{vetterli2002sampling,dragotti2007sampling,maravic2005sampling}.
Therefore,   theoretical results from the FRI sampling theory can be used \cite{vetterli2002sampling,dragotti2007sampling,maravic2005sampling}. In particular,
 the FRI sampling theory tells us that we can find a minimum length {\em annihilating filter}  $\hat h[n]$ such that
\begin{eqnarray}\label{eq:zero}
(\hat h\ast \hat y)[n] =  \sum_{l=0}^k \hat h[l]\hat y[n-l]  = 0, \quad \forall n .
\end{eqnarray}
The specific form of the {\em minimum length} annihilating filter $\hat h[n]$  for the case of \eqref{eq:spec} can be found
in  \cite{vetterli2002sampling}, which has the following z-transform
\begin{eqnarray}\label{eq:afilter}
\hat h(z)  &=& \sum_{l=0}^k \hat h[l] z^{-l} = \prod_{j=0}^{k-1} (1- e^{-i2\pi t_j/\tau} z^{-1}) \ ,
\end{eqnarray}
whose filter length is given by $k+1$  \cite{vetterli2002sampling}.

Now,  one of the most important advances of the proposed method over the existing sampling theory of FRI signals is that,  rather than using the minimum length  annihilating filter,
we allow a longer length annihilating filter which is essential to make the associated Hankel matrix low-ranked.
Specifically, 
if $\hat h[n]$ is a minimum length annihilating filter with $k+1$ filter taps, then for any $k_1\geq 1$ tap filter $\hat a[n]$, it is easy to see that the following filter with $\kappa= k+k_1$ taps is also an annihilating filter for $\hat y[n]$:
\begin{eqnarray}
 \hat h_a[n] = (\hat a \ast \hat h ) [n] \  .
\end{eqnarray}
Accordingly, 
the matrix representation of $(\hat h_a \ast \hat y)[n]=0$ is given by
$$\conv(\hat \yb) \bar{\hat \hb}_a = \zerob$$
where 
 $\bar{\hat\hb}_a$ denotes 
a vector that reverses  the order of the elements  in
\begin{equation}\label{eq:ha}
\hat\hb_a=\left[\hat h_a[0],\cdots, \hat h_a[\kappa-1]\right]^T,
\end{equation}
and
\begin{eqnarray}\label{eq:hankorg}
\conv(\hat \yb) =\left[
        \begin{array}{cccc}
        \vdots & \vdots & \ddots & \vdots \\
        \hat y[-1]  & \hat y[0] & \cdots   & \hat y[\kappa-2]   \\ \hline
      \hat y[0]  & \hat y[1] & \cdots   & \hat y[\kappa-1]   \\
     \hat y[1]  & \hat y[2] & \cdots &   \hat y[\kappa] \\
         \vdots    & \vdots     &  \ddots    & \vdots    \\
      \hat y[n_1-\kappa]  & \hat y[n_1-\kappa+1] & \cdots & \hat y[n-1]\\ \hline
            \hat y[n_1-\kappa+1]  & \hat y[n_1-\kappa+2] & \cdots & \hat y[n_1]\\
              \vdots & \vdots & \ddots & \vdots \\
        \end{array}
    \right]
 \end{eqnarray}
Accordingly,  by choosing $n_1$ such that $n_1-\kappa+1\geq k$ and
 removing the boundary data outside of the  sample indices $[0,\cdots, n_1-1]$, we can construct the following matrix equation:
\begin{eqnarray}
 \hank(\hat \yb) \bar{\hat \hb}_a = \mathbf{0},
\end{eqnarray}
where the Hankel structure matrix $\hank(\hat \yb)  \in \Cd^{(n_1-\kappa+1)\times \kappa} $ is constructed as
    \begin{eqnarray}\label{eq:X2}
\hank(\hat \yb) =\left[
        \begin{array}{cccc}
      \hat y[0]  & \hat y[1] & \cdots   & \hat y[\kappa-1]   \\
     \hat y[1]  & \hat y[2] & \cdots &   \hat y[\kappa] \\
         \vdots    & \vdots     &  \ddots    & \vdots    \\
      \hat y[n_1-\kappa]  & \hat y[n_1-\kappa+1] & \cdots & \hat y[n_1-1]\\
        \end{array}
    \right] 
    \end{eqnarray}
 Then, we have the following result:
\begin{proposition}\label{prp:duality}
Let $k$ denotes the number of Diracs within the support $[0,n_1]$. Suppose, furthermore,  the annihilating filter length $\kappa$ is given by $n_1-\kappa+1\geq k$. Then,
for a given  Hankel structured matrix $\hank(\hat \yb)$ in \eqref{eq:X2}, we have
\begin{eqnarray}
\rank \hank(\hat \yb) \leq  k,  
\end{eqnarray} 
where $\rank(\cdot)$ denotes a matrix rank.
\end{proposition}
\begin{proof}
Note that  \eqref{eq:ha} can be represented as
\begin{eqnarray}
\hat\hb_a  &=&  \conv(\hat \hb) \hat \ab 
\end{eqnarray}
where $\hat\ab =\begin{bmatrix} \hat a[0] & \cdots & \hat a[k_1-1]\end{bmatrix}$ and 
$\conv(\hat \hb) \in \Cd^{\kappa \times k_1}$ is a Toeplitz structured convolution matrix from $\hat \hb$:
\begin{eqnarray}
\conv(\hat \hb) =\begin{bmatrix} \hat h[0] & 0 & \cdots & 0 \\ \hat h[1]  & \hat h[0] & \cdots & 0 \\ \vdots &  \vdots & \ddots & \vdots  \\
\hat h[k-1] & \hat h[k-2] & \cdots & \hat h[k-k_1]  \\  \vdots & \vdots & \ddots & \vdots \\ 0 & 0 & \cdots &  \hat h[k-1] \end{bmatrix}  \in \Cd^{\kappa \times k_1 }
 \end{eqnarray}
where $\kappa=k+k_1$. Since $\conv(\hat \hb)$ is a convolution matrix, it is full ranked and 
we can show that  
$$\dim \conv(\hat \hb)=   k_1,$$
where $\dim(\cdot)$ denotes the dimension of a matrix.
Moreover, the range space of $\conv(\hat \hb)$ now generates the null space of the Hankel matrix, so it is easy to show
$$k_1= \dim \conv(\hat \hb)\leq  \dim \Null \hank(\hat\yb) ,$$
where $\Null(\cdot)$ represent a null space of a matrix.
Hence, we have
\begin{eqnarray}
\rank \hank(\hat \yb)=  \min\{\kappa, n_1-\kappa_1+1\} - \dim \Null \hank(\hat\yb) \leq \kappa - k_1 = k . 
\end{eqnarray}
Q.E.D. 
\end{proof}

This rank condition can be tightened under more restricted set-ups.
Specifically, if $n_1 -\kappa +1 \geq k+1$,   the authors in
\cite{chen2014robust} derived the following decomposition:
\begin{eqnarray}\label{eq:van}
 \hank(\hat \yb)=  L~ D~ R^T   \quad , 
\end{eqnarray} 
where
\begin{eqnarray}\label{eq:L}
L  &=& \begin{bmatrix} 1 & 1 & \cdots & 1 \\ z_0 & z_1 & \cdots & z_{k-1} \\ \vdots & \vdots & \ddots & \vdots \\ z_0^{n_1-\kappa} & z_1^{n_1-\kappa} & \cdots &
z_{k-1}^{n_1-\kappa} \end{bmatrix} \in \Cd^{(n_1-\kappa+1)\times k} \\
R  &=& \begin{bmatrix} 1 & 1 & \cdots & 1 \\ z_0 & z_1 & \cdots & z_{k-1} \\ \vdots & \vdots & \ddots & \vdots \\ z_0^{\kappa-1} & z_1^{\kappa-1} & \cdots &
z_{k-1}^{\kappa-1} \end{bmatrix} \in \Cd^{\kappa \times k} \label{eq:R}
\end{eqnarray}
and
\begin{eqnarray}
D &=&   \begin{bmatrix} c_0 &  0 & \cdots & 0 \\ 0 & c_1 & \cdots & 0 \\ \vdots & \vdots & \ddots & \vdots \\ 0 & 0 & \cdots & c_{k-1} \end{bmatrix},  \label{eq:D}
\end{eqnarray}
and
$z_j = e^{-i2\pi t_j/n_1 }$ for $j=0,\cdots, k-1$.
Because 
$L$ and $R$ are Vandermonde matrices that are full column-ranked, in this case we can further show that
$$\rank\hank(\hat \yb) = \rank(D) = k .$$

However,  as shown in \cite{ye2015compressive}, 
Proposition~\ref{prp:duality} covers more general classes of Hankel matrices that do not necessarily have the Vandermonde decomposition structure in \eqref{eq:van}.
For example,  Proposition~\ref{prp:duality}  is still valid even when the underlying signal is stream of differentiated Diracs,
in which case the conventional low-rank argument in   \cite{chen2014robust} is not valid because it  does not result in a Vandermonde decomposition structure of the associated
Hankel matrix \cite{ye2015compressive}.  Therefore,  Proposition~\ref{prp:duality} will be used   throughout the paper to justify the low-rankness.

It is important to emphasize that Proposition~\ref{prp:duality} implies the following fundamental duality: 
$$\mbox{sparse signal}   \overset{\mathcal{F}} \Longleftrightarrow  \mbox{low-ranked Hankel  structured matrix},$$
where $\mathcal{F}$ denotes the Fourier transform. Therefore, if some of the Hankel matrix elements are missing and one can measure samples only on the index set $\Omega$,  the fundamental duality suggests a way
to recover these elements, which is based on low rank matrix completion \cite{candes2009exact,cai2010singular,candes2010power,gross2011recovering,keshavan2010matrix}:
\begin{eqnarray}\label{eq:EMaC2}
(P)
 &\min_{\mb\in \Cd^{n} } & \rank \hank (\mb)  \\
&\mbox{subject to } & P_\Omega(\mb) = P_\Omega(\hat \lb \odot \hat \fb) \nonumber  \  ,
\end{eqnarray}
where $\odot$ denotes the Hadamard product, and $\hat\lb$ and $\hat \fb$ denotes the vectors composed of discrete samples of $\hat l(\omega)$ and $\hat f(\omega)$, respectively.
Note that by solving $(P)$ we can obtain  the missing data $m(\omega)= \hat l(\omega) \hat f(\omega)$ in the Fourier domain.
Then, the missing  spectral data $\hat f(\omega)$ can be obtained by dividing by the weight, i.e. $\hat f(\omega) =  m(\omega)/\hat l(\omega)$, when $\hat l(\omega) \neq 0$.
As for the  signal $\hat f(\omega)$ at the spectral null of the filter $\hat l(\omega)$,  the corresponding elements should be specifically obtained as sampled measurements,
which can be easily done in MR acquisition.  

Among various algorithm to solve  matrix completion problem $(P)$,   one of the most well-characterised approaches is a convex relaxation approach using the nuclear norm \cite{candes2009exact,cai2010singular,gross2011recovering,keshavan2010matrix}.
More specifically,   the missing k-space elements can be found by solving the following nuclear norm minimization problem:
\begin{eqnarray}\label{eq:EMaC}
(P1) \quad & \min_{\mb\in \Cd^{n_1} } & \|\hank(\mb)\|_* \\
&\mbox{subject to } & P_\Omega(\mb) = P_\Omega(\hat \lb \odot \hat \fb)\nonumber
\end{eqnarray}
where $\|\cdot\|_*$ denotes the matrix nuclear norm. 
In this case,  we can further use the result by Chen and Chi  \cite{chen2014robust} to provide a performance guarantee.
%
%
\begin{theorem}\cite{chen2014robust}
Suppose that the sampling index set $\Omega$ is chosen randomly among $[0, n_1-1]$ and $|\Omega|=m$.
Then, if the number of k-space sample $m$ is given by
\begin{equation}\label{eq:sno}
m> c_1 \mu_1 c_s k \log^4(n_1)
\end{equation}
for some constant $c_1$ and an incoherence parameter $\mu_1$, and 
\begin{equation}\label{eq:cs}
c_s = \max\{ n_1/\kappa,  n_1/(n_1-\kappa+1) \} \ ,
\end{equation}
then  the perfect recovery of the missing spectral components  is possible by solving (P1) with a probability exceeding $1-n_1^{-2}$. 
\end{theorem}

Note that the annihilating filter size $\kappa$ corresponds to  the ``matrix pencil parameter''  \cite{hua1990matrix} in the spectral compressed sensing approach by Chen and Chi \cite{chen2014robust}, who
also showed that   the incoherence parameter $\mu_1$ in \eqref{eq:sno} grows to one as the annihilating filter size $\kappa$  increases \cite{chen2014robust}. 
Hence,  the sampling rate in \eqref{eq:sno} is nearly optimal since it is  proportional to the unknown sparsity level up to $\log^4(\cdot)$ factor.
This suggests an  important observation:  by reformulating the compressed sensing problem as a low rank Hankel structured matrix completion problem in the measurement domain,  no performance loss is expected.

\section{ALOHA for Accelerated MRI}
\label{sec:mri}

Inspired by the theoretical finding in the previous section, this section will explain two realizations of the low-rank matrix completion approaches  that are useful for MR applications. 
The first one is a pyramidal decomposition algorithm that can be used for wavelet domain sparse signals, and the second one is a generalization
for multichannel parallel MRI. 

\subsection{Recovery of Wavelet Sparse Signals}
In fact, the signal $f$ in \eqref{eq:ssp} can be sparsified using wavelet transform as demonstrated in \cite{unser2014unified,unser2014unified2}.
Since the sparsity in wavelet domain have been the main interest in the existing compressed sensing theory, we are particularly interested in this analysis approach. 
Interestingly, the resulting ALOHA framework can have very unique pyramidal structure that  allows   computational efficient  and noise robust implementation of a low rank Hankel matrix completion algorithm.

First,   we assume that we have some real-valued ``L-compatible'' generalized wavelet \cite{unser2014unified,unser2014unified2} which, at a given scale  $s$,  are
defined as $\psi_s(x)$:
\begin{equation}\label{eq:psi}
\psi_s(x) = \mathrm{L}^* \phi_s(x) \ .
\end{equation}
Here, $\mathrm{L}^*$ is the adjoint operator of $\mathrm L$ and $\phi_s$ is a some smoothing kernel at $s$-scale. 
In dyadic wavelet decomposition,  the $s$-scale wavelet is given by $\psi_s(x) = 2^{-s/2} \psi_o(x/2^s)$,
where $\psi_0(x)$ is often called the ``mother  wavelet'' \cite{mallat1999wavelet}.
We further assume that the $\mathrm{L}$-spline model in \eqref{eq:ssp} is cardinal, i.e.
 the knots positions $\{x_j\}_{j=0}^{k-1}$ of the sparse driving signal \eqref{eq:w} are  on integer grid.

Then, using \eqref{eq:psi}, the   wavelet transform of the signal $f(x)$  in \eqref{eq:ssp} at the $s$-scale is given by
\begin{eqnarray}
 \langle f, \psi_s(\cdot - x) \rangle  
 &=& \langle f, \mathrm{L}^* \phi_s(\cdot - x) \rangle \nonumber\\
&=& \langle \mathrm{L} f, \phi_s(\cdot -x) \rangle \nonumber\\
&=& \langle w, \phi_s(\cdot-x) \rangle = (\overline{\phi}_s \ast w)(x) \nonumber\\
&=& \sum_{j=0}^{k-1} c_j \overline{\phi}_s(x-x_j) \label{eq:sample0}
\end{eqnarray}
where $\overline{\phi}_s(x) = \phi_s(-x)$ is the reversed version of $\phi_s$, and the last equality comes from the definition of $w$ in  \eqref{eq:w}.
For the case of ``wavelet sparse'' signals, their  discrete wavelet transform (DWT) coefficients, which corresponds to the discrete samples of \eqref{eq:sample0} with the unit sampling interval,  
 should be non-zero on  sparse index set. 
Hence, if we choose maximally localised smoothing function $ \overline{\phi}_s$, then 
 \eqref{eq:sample0} can be made as sparse as possible. 
 Ideally, if the $\overline{\phi}_0(x)$ has the support size of 1,  the sparsity of the $0$-th scale DWT coefficients is equal to $k$, which achieves the minimum sparsity level.
 In fact, this can be achieved 
using the triangle basis function \cite{unser2010introduction}: 
\begin{eqnarray}\label{eq:p0}
\phi_0(x) =  \beta_+^1(2x), \quad  \mbox{where} \quad \beta_+^1(x) =  \left\{ \begin{array}{ll} x, &   0\leq x < 1  \\ 2-x, &  1 \leq x < 2 \\ 0, & \mbox{otherwise}\end{array} \right.
\end{eqnarray}
which has the support size of  1; hence, it retains the sparsity level of the DWT coefficient to be same as $k$, which corresponds to the
number of Diracs  in the underlying sparse driving signal $w(x)$.  
Furthermore,   for the case of TV signals  i.e. $\mathrm{L}=\frac{d}{dx}$, the corresponding wavelet transform  is indeed Haar wavelet, because
we have 
\begin{eqnarray}
\psi_0(x) = \mathrm{L}^*\phi_0(x) 
=  \left\{ \begin{array}{cl} 1, &  0 \leq  x <  \frac{1}{2}       \\  -1, &  \frac{1}{2}     \leq x <   1 \\  0, & \mbox{otherwise}\end{array} \right. \ .
\end{eqnarray} 

In general, the support size of $\overline{\phi}_s(x)$  with respect to $\overline{\phi}_0(x)$ in \eqref{eq:p0}
 increases to $2^s$, so the maximum sparsity level at the $s$ scale is 
up to $2^sk$. 
To maintain the same sparsity level across the scale, we therefore use the dyadic wavelet decomposition, whose  DWT coefficients are downsampled as 
\begin{equation}\label{eq:sample}
f^s[l] := \left.  \langle f, \psi_s(\cdot - x) 
\rangle
\right|_{x=2^sl} . 
\end{equation} 
Due to the decimation by $2^s$,  it is easy to see that the resulting sparsity level is now  reduced to  $k$.
Then,  the corresponding  Fourier spectrum for the down-sampled signal is given by
\begin{eqnarray}
\hat f^s(e^{i\omega})  :=&\sum_{l=0}^k f^s[l]e^{-i\omega 2^s l} \label{eq:wdec}\\
 =&
\frac{1}{2^s} \sum_{n=0}^{2^s-1} \hat \psi_s^*\left(\omega+n\frac{2\pi}{2^s} \right) \hat f\left(\omega+n\frac{2\pi}{2^s} \right) \label{eq:wavelet_11} \\
=& \frac{1}{\sqrt{2^{s}}} \sum_{n=0}^{2^s-1}\hat \psi^*(2^s \omega+2\pi n) \hat f\left(\omega+n\frac{2\pi}{2^s} \right),   \notag
\end{eqnarray}
where we use the scale dependent wavelet property \cite{mallat1999wavelet}:
$$\hat \psi_s^*(\omega) = \sqrt{2^s}\hat \psi^*(2^s \omega )$$
 and the righthand side terms of \eqref{eq:wavelet_11} come from spectral copies due to the downsampling.

 \begin{figure}[!hbt]
\center{
\includegraphics[width=17cm]{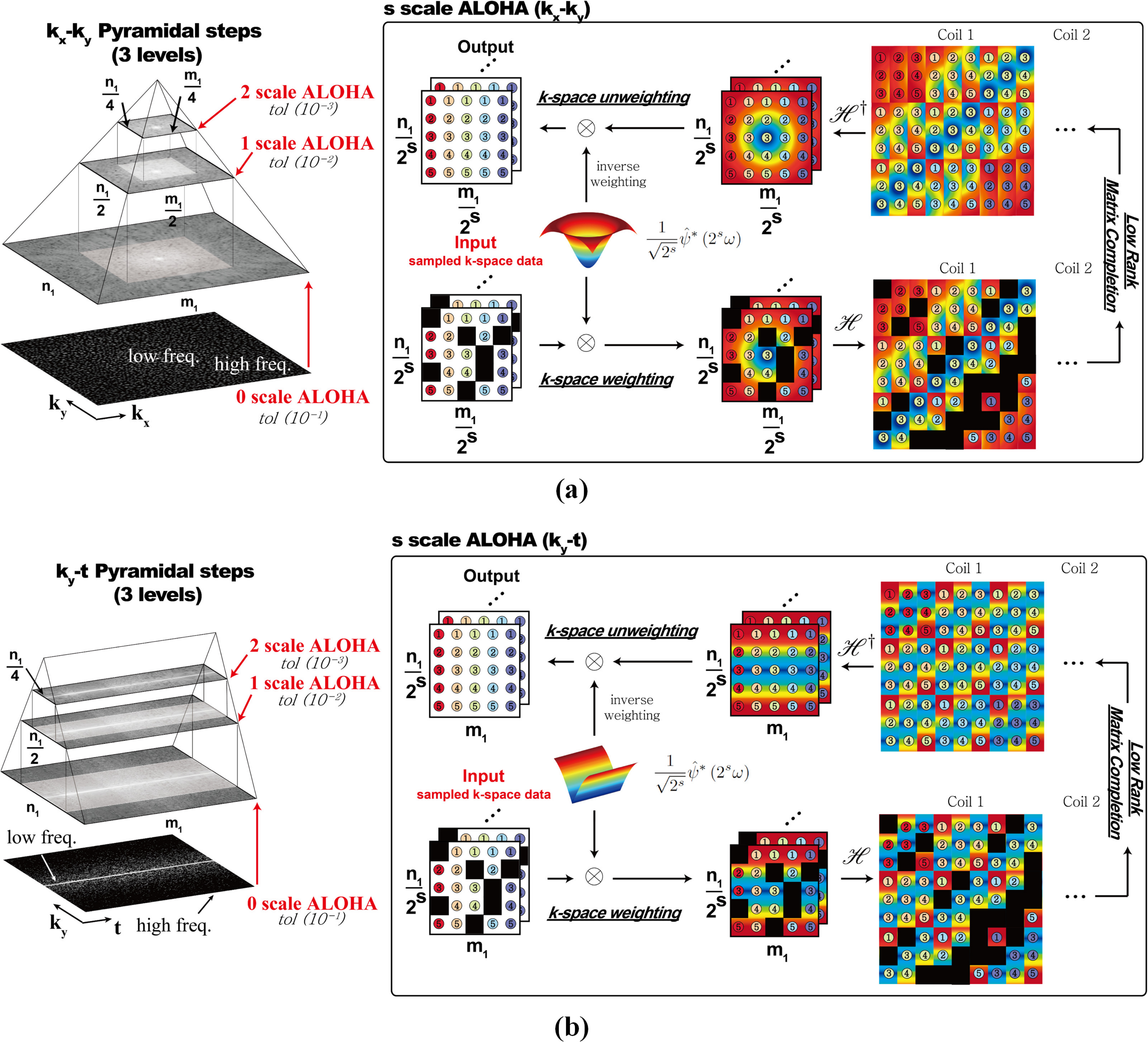}
}
\caption{ALOHA  implementation using pyramidal decomposition.  Construction of Hankel matrices from (a) $k_x$-$k_y$ data by assuming that
2-D dydadic wavelet transform of images is sparse, and (b) $k-t$  subsampled data by assuming that dynamic images can be sparsified using spatial wavelet and temporal
Fourier transform.  
In the box, each reconstruction unit
at the $s$-scale (the $s$-scale ALOHA) is illustrated, which consists of  five steps: (1) the k-space  region extraction,  (2) k-space element by element weighting using $s$-scale
wavelet spectrum  $\frac{1}{\sqrt{2^s}}\hat\psi(2^s\omega)$,  (3) low-rank Hankel matrix completion, (4) k-space unweighting by dividing the interpolated k-space data using the $s$-scale wavelet spectrum, and
 (5) the k-space data replacement using the interpolated data.
 Note that for the case of dynamic MRI, one dimensional weigting is required along the phase encoding direction, whereas 2-D weighting is necessary for
 the case of static imaging. 
 In the figure, $n_1$ denotes the ambient space dimension, and $\hank^\dag$ corresponds to the pseudo-inverse operation 
that takes the average value from the Hankel matrix and putting it back
to the original k-space domain.
The color coding in the Hankel structure matrix indicates the values of weighting.
 }
\label{fig:flowchart_kxky}
\end{figure}

Because the number of nonzero coefficient $f^s[l]$ is at most $k$,  then the form of the righthand side equation \eqref{eq:wdec} is similar to \eqref{eq:spec}, so annihilating filters and the corresponding low-rank Hankel matrix can be found. 
However, one technical issue is that the righthand side of \eqref{eq:wavelet_11} has aliased copies of k-space data weighted by wavelet weighting. Therefore, the direct application of low rank matrix completion may not work.
Accordingly,  to deal with  $2\pi/2^s$ repeating spectral structure, we propose a {\em pyramidal} decomposition of the Hankel structured matrix. 
More specifically, 
%
if the spectral component of $\hat f(\omega)$ is estimated  for  $|\omega| \geq \pi/2^{s}$, then we can constitute a residual
signal recursively:
\begin{eqnarray}
\hat f_{r}(\omega) :=  \left\{ \begin{array}{ll} 0, &|\omega| \geq \pi/2^{s} \\  \hat f(\omega), & \mbox{otherwise} \end{array} \right.
\end{eqnarray}
Then,  for the corresponding discrete sample
$f_{r}^{s}[l]:=\left. \langle f_r, \psi_s(\cdot - x) \rangle\right|_{x=2^{s}l}$ 
with $f_{r} = \mathcal{F}^{-1}\{\hat f_{r} \}$, we have 
\begin{eqnarray}\label{eq:wdec3}
\hat f_r^{s}(e^{i\omega}):= 
 \sum_{l}   f_{r}^{s}[l]e^{-i\omega 2^{s}l}  
&=&
 \frac{1}{\sqrt{2^{s}}} \hat \psi^*(2^{s} \omega) \hat f (\omega),
\end{eqnarray}
because its high frequency content will  no more aliasing at a frequency inside $|\omega| \leq \pi/2^{s}$.
Then, due to the sparsity of $f_{r}^{s}[l]$, there exists a filter that annihilates $\hat \psi^*(2^{s} \omega) \hat f (\omega)$.
Consequently, the $s$-scale ALOHA problem can be formulated as
\begin{eqnarray}\label{eq:EMaC3}
&\min_{\mb\in \Cd^{n_1/2^s} } & \rank \hank(\mb)  \\
&\mbox{subject to } & P_\Omega(\mb) = P_\Omega ( \hat \lb^s \odot \hat \fb^s) \nonumber  \  .
\end{eqnarray}
where 
\begin{eqnarray}\label{eq:xs}
\hat f^s[n] =  \left. \hat f(\omega) \right|_{\omega = n \pi/n_1},  
\end{eqnarray}
and 
 the weighting vector  $\hat\lb^s$ is composed of following spectral samples :
\begin{eqnarray}
\hat l^s[n] &= &\left. \hat l^s(\omega) \right|_{\omega = n \pi/n_1}  
= \left. \frac{1}{\sqrt{2^s}} \hat \psi^*\left(2^s \omega\right)  \right|_{\omega = n \pi/n_1}, \label{eq:ls}
\end{eqnarray}
where  $n= -{n_1}/{ 2^{s+1}},\cdots,  {n_1}/{2^{s+1}} -1$ and  the superscript $^*$ denotes the complex conjugate.
Note that sampling interval for $\hat\psi^*(\omega)$  widens by factor of 2 for each successive scale to construct the  k-space weighting vector $\hat \lb^s$.
Consequently, we have a  pyramidal decomposition as shown in Fig.~\ref{fig:flowchart_kxky}, and  the  ALOHA \eqref{eq:EMaC3} should be solved from the  lowest scale, i.e. $s=0$, up to the highest scale (see caption of   Fig.~\ref{fig:flowchart_kxky} for more discussion of the figures).
Because  sparsity can be imposed only on the wavelet coefficients,  the low frequency k-space data that correspond to the scaling function coefficients should be acquired additionally during MR data acquisition.
This information as well as the annihilating filter size then
 determines the the depth of the pyramidal decomposition as will be discussed  in detail later.

There are several advantages of using wavelet approaches compared to the direct operator weighting in ALOHA framework. First,  the pyramidal decomposition structure 
for low rank matrix completion can significantly reduce the overall  computational complexity. Moreover, it has been observed in \cite{unser2010introduction} that the wavelet approach is more
robust for noise and the model mismatch, which is also consistently observed in ALOHA framework as will be discussed later. 

Note that  the proposed pyramidal scheme is an algorithm that is designed to approximately decouple the aliasing components
in \eqref{eq:wavelet_11}, and  we do not claim the optimality of the proposed method.  
However, it is worth to mention that the optimality of a similar fine-to-coarse scale wavelet coefficients reconstruction was recently shown in
Fourier compressed sensing problem \cite{zhang2015reconstruction}.



\subsection{Generalization to Parallel MRI} 	

Beside  the  annihilation property originating from  sparsity in the transform domain,  there exists an additional annihilation relationship that is unique in parallel MRI.
The relationship we described here has similarity to  SAKE and P-LORAKS when the image itself is sparse (even though these methods were inspired by the null space or GRAPPA type relationship), 
and  our contribution is its generalization to transform domain sparse signals.

Specifically, in pMRI, the unknown image $g_i(x)$ from the $i$-th coil  can be represented as
$$g_i(x) = s_i(x) f(x),\quad i=1,\cdots, r, $$
where $s_i(x)$ denotes the $i$-th coil sensitivity map, $f(x)$ is an unknown image, and $r$ denotes the number of coils.
 Then, for TV signals  
(i.e. $\mathrm{L}=\frac{d}{dx}$),  
 we have
\begin{eqnarray*}
\mathrm L g_i(x) 
&=&  s_i(x)\mathrm Lf(x) +   f(x) \mathrm{L} s_i(x) \ , 
\end{eqnarray*}
whose Fourier transform is given by
\begin{eqnarray}\label{eq:lg}
\hat l(\omega) \hat g_i(\omega) 
&=&  \hat s_i(\omega) \ast (\hat l(\omega) \hat f(\omega) )+  \hat f(\omega) \ast (\hat l(\omega) \hat s_i(\omega)) \ . 
\end{eqnarray}
Then, we can show the following inter-coil annihilating filter relationship:
\begin{lemma}
Suppose that $\hat h(\omega)$ is the annihilating filter for $\hat f(\omega)$ such that $(\hat h\ast \hat f)(\omega) =0$. Then, we have
\begin{eqnarray}\label{eq:an_mat2}
 \hat h(\omega)\ast \hat s_j(\omega) \ast \left(\hat l(\omega)\hat g_i(\omega)\right) - \hat h(\omega)\ast \hat s_i(\omega) \ast  \left(\hat l(\omega)\hat g_j(\omega)\right) = 0, \quad  i\neq j. 
\end{eqnarray}
\end{lemma}
\begin{proof}
Using \eqref{eq:lg}, we have
\begin{eqnarray}
\hat h(\omega)\ast \hat s_j(\omega) \ast \left(\hat l(\omega)\hat g_i(\omega)\right)  &=& \hat h(\omega)\ast \hat s_j(\omega) \ast \hat s_i(\omega) \ast (\hat l(\omega) \hat f(\omega) )+ \hat h(\omega)\ast \hat s_j(\omega) \ast \hat f(\omega) \ast (\hat l(\omega) \hat s_i(\omega))  \notag \\
&=&  \hat h(\omega)\ast \hat s_j(\omega) \ast \hat s_i(\omega) \ast (\hat l(\omega) \hat f(\omega) ) \label{eq:tv1}
\end{eqnarray}
where the second term vanishes thanks to  $(\hat h\ast \hat f)(\omega) =0$.  Similarly, we have
\begin{eqnarray}
\hat h(\omega)\ast \hat s_i(\omega) \ast \left(\hat l(\omega)\hat g_j(\omega)\right)  &=&  \hat h(\omega)\ast \hat s_i(\omega) \ast \hat s_j(\omega) \ast (\hat l(\omega) \hat f(\omega) ) \label{eq:tv2}
\end{eqnarray}
Because Eqs.~\eqref{eq:tv1} and \eqref{eq:tv2} are identical, Eq.~\eqref{eq:an_mat2} can be shown.  
Q.E.D.
\end{proof}
Therefore,    to exploit the inter-coil annihilating property,  Hankel matrix from each channel is stacked side by side to form a matrix $\Yc$:
\begin{eqnarray}\label{eq:Yc}
\Yc = \begin{bmatrix}  \hank(\hat \lb \odot \hat \gb_1) & \cdots &  \hank(\hat \lb \odot \hat \gb_r) \end{bmatrix}   \in \Cd^{(n_1-\kappa+1)\times \kappa r}\  ,
\end{eqnarray}
Here, the augmented matrix $\Yc$ in \eqref{eq:Yc} has the following rank condition:
 \begin{proposition}\label{prp:prank}
 For the concatenated Hankel matrix given in \eqref{eq:Yc}, we have 
 \begin{eqnarray}
\mathrm{rank}~ \Yc  &\leq & \frac{r(2k-r+1)}{2}  \ .
\end{eqnarray}
where $k$ and $r$ denotes the sparsity and the number of coils, respectively.
 \end{proposition}
 \begin{proof}
Since 
Eq.~\eqref{eq:an_mat2} holds for every pair among $r$-channels,  it is easy to show that
$$\Yc \Sc_1 =\zerob,$$ 
where $\Sc_1$ is defined recursively as follows:
\begin{eqnarray}
  \Sc_{r-1} &\triangleq& \left[
                  \begin{array}{cc}
                    \hat\vb_r \\    - \hat\vb_{r-1}  \\
                  \end{array}
                \right]
   \\
  \Sc_t &\triangleq&\left[
                      \begin{array}{cccc|c}
                       \hat\vb_{t+1} & \hat\vb_{t+2} & \cdots & \hat\vb_{C}  & \zerob \\ \hline
                       -\hat\vb_t &  & && \\
                         & -\hat\vb_t &  &  & \Sc_{t+1}  \\
                          &  & \ddots &  \\
                         &  &  &   -\hat\vb_t  & 
                      \end{array}
                    \right]  \ ,
 \end{eqnarray}
 wbere  $\hat \vb_i$ denotes the discrete samples of $\hat h(\omega) \ast \hat s_i(\omega)$.
Accordingly,
$$ \dim \Null(\Yc)\geq \mathrm{rank}(\Sc_1) = \binom{r}{2} = {r(r-1)}/{2}. $$ 
Furthermore, Proposition~\ref{prp:duality} informs us that 
$ \hank(\hat \lb \odot \hat \gb_i)$ in $\Yc$ has rank at most $k$. 
Therefore, we have
\begin{eqnarray}
\mathrm{rank}~ \Yc  &\leq&  k r -  \frac{r(r-1)}{2} 
=
\frac{r(2k-r+1)}{2}  \ .
\end{eqnarray}
\end{proof}

Due to the low-rankness of the concatenated matrix,  
 the multichannel version of the ALOHA can be formulated as
\begin{eqnarray}\label{eq:EMaC4}
 &\min_{ \{\mb_i\}_{i=1}^r } & \rank \begin{bmatrix}  \hank( \mb_1) & \cdots &  \hank(\mb_r) \end{bmatrix}  \\
&\mbox{subject to } & P_\Omega(\mb_i) = P_\Omega(\hat \lb \odot \hat \gb_i) ,\quad i=1,\cdots, r \nonumber  \  .
\end{eqnarray}

In addition, it is worth to emphasize the meaning of the rank condition of the concatenated matrix.
 Because  the rank of $\Yc$ is  at most ${r(2k-r+1)}/{2}$ and there exist
additional degrees of freedom that belong to  the amplitudes of  Diracs in the transform domain, 
the total degrees of the freedom for the FRI signal are at most $r(2k-r+1)$  \cite{vetterli2002sampling,dragotti2007sampling,maravic2005sampling}.
In parallel MRI,  if $m$ denotes the number of k-space sampling locations,
then k-space data are sampled simultaneously from 
 $r$-coils and
the total number of  k-space samples becomes $mr$. Since the number of the samples should be larger than the degree of the freedom, we have
\begin{eqnarray}\label{eq:bound}
mr \geq {r(2k-r+1)}  & \Longleftrightarrow & k \leq  \frac{m+r-1}{2} \ .
\end{eqnarray}
This result has very important geometric meaning.
Suppose that we attempt to address parallel imaging by exploiting the joint sparsity. Using the notation in \eqref{eq:sense}, this could be addressed by the following multiple measurement vector (MMV) problem
\begin{eqnarray}
\min_F &\|F\|_0 \nonumber\\
\mbox{subject to} & G=AF
\end{eqnarray}
where $G=[\gb_1,\cdots,\gb_r]$ and $F=[\fb_1,\cdots, \fb_r]$ with $\fb_i=[S_i] \fb$ for a given sensitivity map $[S_i]$,
and $\|\cdot\|_0$ denotes the number of non-zero rows.
Then,  it was shown in our previous work  \cite{kim2012compressive} and others \cite{lee2012subspace,davies2012rank} that if $F_*$ satisfies $AF_*=G$ and 
\begin{eqnarray}\label{eq:mmv}
\|F_*\|_0 <  \frac{\mathrm{spark}(A)+ \rank(G) -1 }{2} 
\end{eqnarray}
then $F_*$ is the unique solution for $G=AF$, where $\mathrm{spark}(A)$ denote the smallest number of linearly dependent columns of $A$.
For randomly chosen Fourier samples, we have $\mathrm{spark}(A)\geq m+1$. Moreover,  $ \rank(G) =r$. Hence,
\eqref{eq:mmv} is equivalent to \eqref{eq:bound} if the unknown sparsity level is $k$, i.e. $\|F_*\|_0=k$.

Therefore, similar to the fact that single coil ALOHA does not lose any theoretical optimality compared to the single coil compressed sensing approaches,
our ALOHA formulation for  pMRI using 
 \eqref{eq:EMaC4} may fully exploit the multi-channel diversity from parallel acquisition.
Moreover, because the k-space interpolation formulation of ALOHA is derived based on the information of the singularity of the signals, the resulting reconstruction
results for both single and multi-channel formulation exhibit superior performance as will be discussed in  experimental sections.

\section{Implementation Details}
\label{sec:method}

\subsection{2D Hankel Structured Matrix Construction}

\begin{figure}[!hbt]
\centering
\includegraphics[trim = 0mm 0mm 0mm 0mm,clip=true,width=17cm]{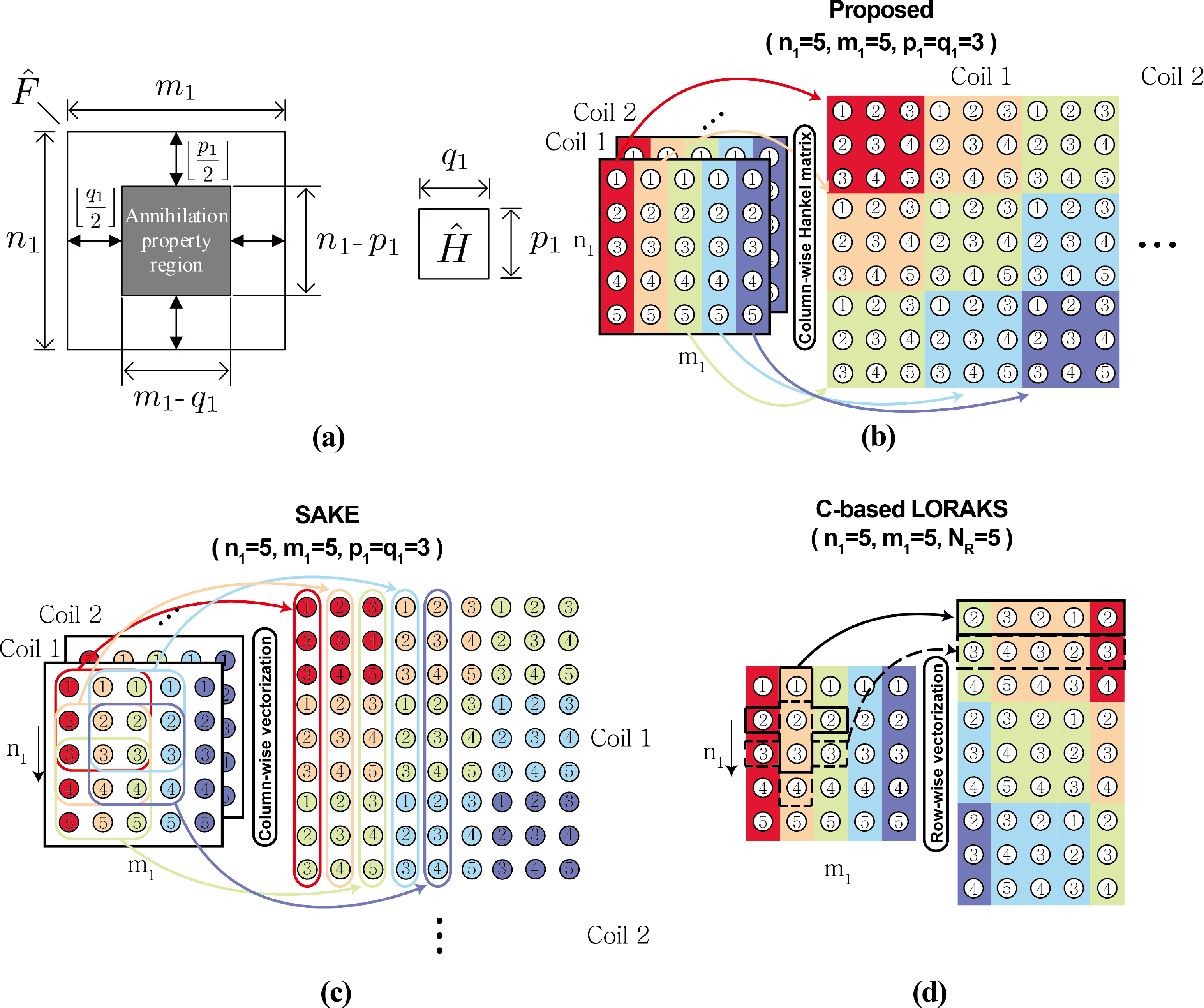}
\caption{(a) Area where annihilation property holds. Various
ways of constructing  block Hankel matrices: (b) ALOHA, (c) SAKE, and (d) LORAKS. In (d), $N_R$ denotes
 the number of neighborhood pixels.}\label{fig:hankel_mat}
\end{figure}

In 3D imaging or dynamic acquisition of MR data, the readout direction is usually fully sampled and the other two  encoding directions are 
under sampled. Thus, this section presents an explicit way of constructing a 2D Hankel structured matrix.
Specifically, if $\hat h[n,m]$ is a $p_1\times q_1$ size 2D annihilating filter,
then the corresponding annihilating filter 
 relation is given by
\begin{eqnarray}\label{eq:zero2}
(\hat h\ast \hat f)[n,m] =  \sum_{i=0}^{p_1-1}\sum_{j=0}^{q_1-1} \hat h[i,j]\hat f[n-i, m-j]  = 0, 
\end{eqnarray}
for all $n, m \in \Omega$.
Let $n_1 \times m_1$ k-space data matrix be defined by
\begin{eqnarray*}
\hat F &:=& \begin{bmatrix}\hat f[0,0] & \cdots & \hat f[0,m_1-1] \\ \vdots & \ddots & \vdots \\ \hat f[n_1-1,0] & \cdots & \hat f[n_1-1, m_1-1] \end{bmatrix}  \\
 &=&
\begin{bmatrix} \hat\fb_0 & \cdots  &\hat \fb_{m_1-1} \end{bmatrix}
\end{eqnarray*}
Similarly, we define $p_1\times q_1$ annihilating filter matrix $\hat H$.
Then, by removing the boundary effect from the 2D convolution as shown in Fig.~\ref{fig:hankel_mat}(a),
the 2D annihilation property should hold only inside of the domain and
 \eqref{eq:zero2}   can be equivalently represented as
\begin{eqnarray}\label{eq:an_final}
 \hank(\hat F)  \hat \hbk = \mathbf{0},
    \end{eqnarray}
where a 2-D Hankel structured matrix $  \hank(\hat F) $ is constructed as
\begin{eqnarray}\label{eq:3dhank}
   \!\left[\!
        \begin{array}{cccc}
            \!
        \hank(\hat \fb_0)&    \hank(\hat \fb_1) &   \cdots &  \hank(\hat \fb_{q_1-1})\\
           \hank(\hat \fb_1)&    \hank(\hat \fb_2) &   \cdots &  \hank(\hat \fb_{q_1})\\
       \vdots &  \vdots  & \ddots  & \vdots \\
   \hank(\hat \fb_{m_1-q_1} )&    \hank(\hat \fb_{m_1-q_1+1}) &   \cdots &  \hank(\hat \fb_{m_1-1})\\
                     \end{array}\!
              \right]
    \end{eqnarray}
with $\hank(\hat \fb_j)   \in \Cd^{(n_1-p_1+1)\times p_1} $  given by
   \begin{eqnarray*}
 \left[
        \begin{array}{cccc}
      \hat f[0,j]  & \hat f[1,j] & \cdots   & \hat f[p_1-1,j]   \\
     \hat f[1,j]  & \hat f[2,j] & \cdots &   \hat f[p_1,j] \\
         \vdots    & \vdots     &  \ddots    & \vdots    \\
      \hat f[n_1-p_1,j]  & \hat f[n_1-p_1+1,j] & \cdots & \hat f[n_1-1,j]\\
        \end{array}
    \right] \ ,
    \end{eqnarray*}
    and the annihilating filter vector is given by
    \begin{eqnarray}
    \hat\hbk =\overline{\vect(\hat H)} \ ,
    \end{eqnarray}
    where the overline denotes an operator that reserves the order of a vector. 
    Using this, we can construct an augmented matrix $\Yc$ in \eqref{eq:Yc} from $r$-channels.
    
The augmented matrix structure   $\hank(\hat F)$ illustrated in Fig.~\ref{fig:hankel_mat}(b)  is similar to those of SAKE and LORAKS/P-LORAKS  in Fig.~\ref{fig:hankel_mat}(c) and (d), respectively, with the following differences. 
Compared to SAKE,   ALOHA stacks the multi-coil Hankel matrices side by side.  Unlike the SAKE and ALOHA,  LORAKS uses 
$R=4$ ($R$: neighborhood size) neighbors according to the software manual provided by the original
authors.
Moreover, the most important novelty of the proposed method is the k-space weighting to construct  the Hankel structured matrix.

\subsection{Hankel structured matrix completion algorithm} \label{sec:reconstruction} 	

In order to solve  Eqs.  \eqref{eq:EMaC2} and \eqref{eq:EMaC4},
we employ an SVD-free structured rank minimization algorithm \cite{signoretto2013svd} with an initialization using the
 low-rank factorization model (LMaFit) algorithm \cite{wen2012solving}.  This algorithm does not use the singular value decomposition (SVD), so the computational complexity can be significantly reduced. 
Specifically, the algorithm is based on the following observation \cite{srebro2004learning}:
 \begin{eqnarray}\label{eq:relaxation_nuclear}
 \|A\|_* = \min\limits_{U,V: A=UV^H} \|U\|_F^2+ \|V\|_F^2 \quad \ .
 \end{eqnarray}
Hence,   \eqref{eq:EMaC2} can be reformulated as 
 the nuclear norm minimization problem under  the matrix factorization constraint:
\begin{eqnarray}
\min_{U,V:\hank(\mb)=UV^H}  && \|U\|_F^2+ \|V\|_F^2 \nonumber \\
\mbox{subject to}
&&P_\Omega(\mb)=  P_\Omega(\hat \fb)
 . \label{eq:data}
\end{eqnarray}
By combining the two constraints, we have the following cost function for   an alternating direction method of multiplier  (ADMM) step \cite{boyd2011distributed}:
\begin{eqnarray}\label{eq:ADMM}
L(U,V,\mb,\Lambda) & := & \iota(\mb) + \frac{1}{2} \left( \|U\|_F^2+\|V\|_F^2\right)  \nonumber\\
&&+  \frac{\mu}{2}  \|\hank(\mb)- UV^H+\Lambda\|^2_F
\end{eqnarray}
where $\iota(\mb)$  denotes an indicator function:
$$\iota(\mb) = \left\{\begin{array}{ll} 0, & \mbox{if $P_\Omega(\mb)=  P_\Omega(\hat \fb)$} \\ \infty, & \mbox{otherwise} \end{array} \right.  \  .$$
One of the advantages of the ADMM formulation is that each subproblem is simply obtained from \eqref{eq:ADMM}. More specifically,
 $\mb^{(n+1)},  U^{(n+1)}$ and  $V^{(n+1)}$ can be obtained, respectively, by applying the following optimization problems sequentially:
 \begin{equation}\label{eq:update}
 \begin{array}{l}
 \min_\mb  \iota(\mb)  +  \frac{\mu}{2}  \|\hank(\mb)- U^{(n)}V^{(n)H}+\Lambda^{(n)}\|^2_F \\
 \min_U  \frac{1}{2} \|U\|_F^2 + \frac{\mu}{2}  \|\hank(\mb^{(n+1)})- UV^{(n)H}+\Lambda^{(n)}\|^2_F \\ 
  \min_V  \frac{1}{2} \|V\|_F^2 + \frac{\mu}{2}  \|\hank(\mb^{(n+1)})- U^{(n+1)}V^{H}+\Lambda^{(n)}\|^2_F 
  \end{array}
  \end{equation}
  and the Lagrangian update is given by
  \begin{eqnarray*}
  \Lambda^{(n+1)} =& \Yc^{(n+1)}- U^{(n+1)}V^{(n+1)H}+\Lambda^{(n)} \ .
 \end{eqnarray*}
 It is easy to show that the first step in \eqref{eq:update} can be reduced to
 \begin{eqnarray}
 \mb^{(n+1)}  = 
 P_{\Omega^c}\hank^{\dag}\left\{  U^{(n)}V^{(n)H}- \Lambda^{(n)} \right\} + P_\Omega (\hat \fb) ,
 \end{eqnarray}
 where $P_{\Omega^c}$ is a projection mapping on the set $\Omega^c$ and
  $\hank^{\dag}$ corresponds to the Penrose-Moore pseudo-inverse mapping from our block Hankel structure to a vector.
Hence, the role of the pseudo-inverse is taking the average value and putting it back
to the original coordinate. 
Next, the subproblem for $U$ and $V$ can be easily calculated by taking the derivative with respect to each matrix, and we have
  \begin{equation}\label{eq:UV}
  \begin{array}{l}
 U^{(n+1)} = \mu \left(\Yc^{(n+1)}+\Lambda^{(n)}\right)V^{(n)} \left(I+\mu V^{(n)H}V^{(n)}\right)^{-1}    \\
  V^{(n+1)} = \mu\left(\Yc^{(n+1)}+\Lambda^{(n)}\right)^HU^{(n+1)} \left(I+\mu U^{(n+1)H}U^{(n+1)}\right)^{-1}   
 \end{array}
 \end{equation}
 Note that the computational complexity of our ADMM algorithm is dependent on the matrix inversion in \eqref{eq:UV},
 whose complexity is determined by the estimated rank of the Hankel matrix. Therefore, even though the Hankel matrix has large size,
 the estimated rank is much smaller, which significantly reduces overall complexity.

  Now, for faster convergence, the remaining issue is how to initialize $U$ and $V$. For this, we employ an algorithm called the low-rank factorization model (LMaFit) \cite{wen2012solving}.
 More specifically, for a low-rank matrix $Z$, LMaFit solves the following optimization problem:
\begin{equation}\label{eq:lmafit}
\min_{U,V,Z} \frac{1}{2} \|UV^H-Z\|^2_F~  \mbox{subject to } P_I(Z) = P_I(\hank(\hat \fb))
\end{equation}
and $Z$ is initialized with $\hank(\hat \fb)$ and the index set $I$ denotes the positions where the elements of $\hank(\hat \fb)$ are known.
LMaFit solves a linear equation with respect to  $U$ and $V$ to find their updates and relaxes the updates by taking the average between the previous iteration and the current iteration.
Moreover, the rank estimation can be done automatically. LMaFit uses QR factorization instead of SVD, so it is also computationally efficient. 
Even though the problem  \eqref{eq:lmafit} is non-convex due to the multiplication of $U$ and $V$,
the convergence of LMaFit to a stationary point was analyzed in detail \cite{wen2012solving}. However, the LMaFit alone cannot recover the block Hankel structure, which is the reason we use an ADMM step afterward to impose the structure.


\subsection{Reconstruction Flow} 	

As shown in Fig.~\ref{fig:flowchart_kxky}, the  ALOHA framework is comprised with several major steps:  pyramidal decomposition,  k-space weighting, Hankel matrix formation, rank estimation,  SVD-free low rank matrix completion, and  k-space unweighting. 
Here, we will explain these in more detail.

The pyramidal decomposition is performed as follows. First,  in  static MR data acquisition illustrated in Fig.~\ref{fig:flowchart_kxky}(a),  
the $k_x-k_y$ corresponds to the two phase encoding directions that are downsampled. Thus, the Hankel matrix is constructed from $k_x-k_y$ data.
After a k-space interpolation from a finer scale,
the data  at the current scale  is defined to contain one-fourth of data around zero frequency  from that of the previous scale.
Second, in the case of dynamic MR imaging shown in Fig. \ref{fig:flowchart_kxky}(b),   $k_x$  samples from the readout direction are fully acquired,
whereas the $k_y$ directional phase encoding are downsampled along the temporal direction $t$.
Therefore, the data in $k_y-t$ space (or simply, $k-t$ space) are downsampled,  from which we construct a Hankel structure matrix.
In pyramidal decomposition,  after a k-space interpolation from a finer scale,  the $k_y-t$ data in the current scale
contains a half of the data from that of the previous scale.   Note that  the wavelet decomposition is performed only along the spatial domain, 
so the pyramidal decomposition is only performed along $k_y$ direction.
This construction of Hankel matrix is due to the observation that the dynamic signal is sparse in spatial wavelet and temporal Fourier transform domain \cite{Lee2015MRPM}. 
See more details  in our recent work \cite{Lee2015MRPM}.
%
%
%

In both cases, the estimated k-space data at the lower scale are used to initialize the low rank matrix completion algorithm at the current scale. This accelerates the convergence speed. Moreover, due to the additional chance of  refining the estimates,
 more important k-space samples at the low frequency regions are refined furthermore compared to the high frequency k-space samples.  Consequently, the overall computational burden of the low rank matrix competition algorithm is significantly reduced while the overall quality is still maintained.


The k-space weighting is performed using wavelets.
Specifically, based on our discussions on maximally localized smoothing function, 
we use a Haar wavelet expansion whose 
spectrum is given by
\begin{eqnarray}\label{eq:wavelet_spline_dyadic}
\hat\psi_0(\omega)=\frac{i\omega}{2}\left(\frac{\sin{\omega/4} }{\omega/4}\right)^2\exp\left( -i\frac{\omega}{2}\right)& \ .
\end{eqnarray}
The corresponding k-space weighting at the $s$-scale is given by
 $$\hat l_s(\omega) =  
\frac{1}{\sqrt{2^s} }\frac{i2^s\omega}{2}\left(\frac{\sin{2^s\omega/4} }{2^s\omega/4}\right)^2e^{- i\frac{2^s\omega}{2}} .$$
For the case of static  MRI in Fig.~\ref{fig:flowchart_kxky}(a),  we use 2-D weighting by assuming that the image is sparse in 2-D dyadic wavelet transform domain.
Care needs to be taken when applying the weighting to 2D Fourier domain because there are two frequency  variables $(\omega_x,\omega_y)$.
One could use a  separable weighting  $\hat l(\omega_x,\omega_y) = \hat l(\omega_x)\hat l(\omega_y)$; however, the resulting problem is  that
the missing k-space components along the frequency axis $\omega_x=0$ or $\omega_y=0$ cannot be recovered. Consequently,
we applied the weighting sequentially along each axis, i.e. we solve \eqref{eq:EMaC2} by applying $\hat l(\omega_x)$ first, which is followed by solving \eqref{eq:EMaC2} with $\hat l(\omega_y)$.
However, simultaneous weighting would be possible as demonstrated in a recent work for off-the-grid recovery of piecewise constant image \cite{ongie2015off}.
For the case of dynamic imaging in Fig.~\ref{fig:flowchart_kxky}(b), one dimensional weighting along the phase encoding direction was applied  as explained in detail
in \cite{Lee2015MRPM}.  
Finally, after the k-space interpolation, the k-space unweighting is done in k-space pixel-by-pixel by dividing the reconstructed value with \eqref{eq:wavelet_spline_dyadic}.
Note that \eqref{eq:wavelet_spline_dyadic} has zero value at the DC frequency. However, because we acquire the DC value as well as some of the low frequency k-space data, the problem of dividing by zero never happened.

 Because LORAKS is similar to the ALOHA without weighting,  we use LORAKS as a reference to contrast
 why the proposed ALOHA framework has many advantages. The implementation of LORAKS was based on the source code available in the original author's homepage,
 which requires manual setting of estimated ranks.  We chose the rank for LORAKS that gave the best reconstruction quality.
 In Discussion, we also provided reconstruction results by ALOHA without weighting to control the confounds and confirm the importance of the k-space weighting
 in constructing Hankel matrix.

We used TITAN GTX graphic card for graphic processor unit (GPU) and i7-4770k CPU and the codes were written in  MATLAB 2015a (Mathwork, Natick).
To accelerate the algorithm, most part of the MATLAB codes were implemented using Compute Unified Device Architecture (CUDA) for GPU.

\subsection{MR Acquisition and Reconstruction Parameters} 	

To assess the performance of ALOHA for single coil compressed sensing imaging,  k-space raw data from an MR headscan was obtained  with Siemens Tim Trio 3T scanner using balanced steady-state free precession (bSSFP) sequence. 
The acquisition parameters were as follows: TR/TE $= 10.68/5.34\mbox{ms}$, $208\times256$ acquisition matrix (partial Fourier factor 7/8, oversampling factor $50\%$), and number of slices is 104 with $2\mbox{mm}$ slice thickness. The field-of-view (FOV) was $178\times220\mbox{mm}^2$. We used the central coronal slice.

A retrospective down-sampling mask was generated according to a two dimensional Gaussian distribution and the data at the central $7\times 7$  region around zero frequency were obtained additionally.  This is equivalent to assume a 3D imaging scenario where the readout direction is fully sampled and the downsampling is done in the remaining 2-D phase encoding direction.
Downsampling factors of four was used to generate sampling masks. However, data was obtained as partial Fourier measurements, so effective downsampling ratio was 4.13. The 2-D k-space weighting using \eqref{eq:wavelet_spline_dyadic} was used. The ALOHA reconstruction was conducted using the following parameters: three levels of pyramidal decomposition, and decreasing LMaFit tolerance values ($5\times10^{-2},5\times10^{-3},5\times10^{-4}$) at each level of the pyramid. In addition, an initial rank estimate for LMaFit started with one and was refined automatically in an increasing sequence, the annihilating filter size was 23$\times$23, and the ADMM parameter was $\mu=10^3$.

%
%

For the  compressed sensing approach, we used two approaches with the same data and the same sampling masks: (1) the sparsity in wavelet domain (which we denote $l_1$-wavelet) \cite{lustig2007sparse}, and (2) the split Bregman method for the total variation \cite{goldstein2009split}. In case of $l_1$-wavelet approach, we implemented an ADMM algorithm using wavelet domain sparsity. In addition, for comparison with the existing state-of-the art approach using Hankel structured matrix completion algorithm, LORAKS was compared. In particular, C-based LORAKS was chosen because it exploits the image domain sparsity.
The parameters for the $l_1$-wavelet and TV approaches were optimized to have the best performance in terms of the normalized mean square error (NMSE), which is defined by
 	$\mbox{NMSE}(\xb)={\|\xb -\yb\|_2^2}/{\|\yb\|_2^2}$,
where $\xb$ and $\yb$ denote the reconstructed and the ground-truth images, respectively. The LORAKS parameters were chosen manually to give the best reconstruction quality.

To evaluate the performance of ALOHA in static parallel imaging,  
 k-space raw data from an MR headscan was obtained  with Siemens Verio 3T scanner using 2D SE sequence. 
The acquisition parameters were as follows: TR/TE $= 4000/100\mbox{ms}$, $256\times256$ acquisition matrix, and six z-slices with $5\mbox{mm}$ slice thickness. The field-of-view (FOV) was $240\times240\mbox{mm}^2$, and the number of coils was four.  
Retrospectively undersampled 2-D k-space data at the acceleration factor of eight were obtained according to a
two dimensional Gaussian distribution in addition to the $7\times7$ central region around zero frequency. The data from 4 receiver coils were used. For comparison, we used the
 identical data and sampling masks for  SAKE \cite{shin2013calibrationless}, and SAKE with ESPIRiT \cite{uecker2014espirit}. Note that  GRAPPA \cite{griswold2002generalized} requires ACS lines, so  with the additional 50 samples along ACS, the effective downsampling ratio was 4.785, which is not good as the eight time acceleration in other algorithms.
SAKE and SAKE with ESPIRiT are both low rank matrix completion algorithms for Hankel structured matrix collected from the whole k-space data.
However, SAKE with ESPIRiT reduces the computational burden of the original SAKE by performing low rank matrix completion only for the $65\times 65$ central region with $7\times 7$ filter, after which
coil sensitivities are estimated using the reconstruction data. The estimated coil sensitivities are used to estimate the remaining k-space missing data through ESPIRiT \cite{uecker2014espirit}.
The parameters for SAKE and SAKE with ESPIRiT were chosen such that they provided the best reconstruction results.
The parameters for the  ALOHA are as follows: three levels of pyramidal decomposition with decreasing LMaFit tolerances ($10^{-1},10^{-2},10^{-3}$), 
 and 7$\times$7 annihilating filter.  The same  LMaFit rank estimation strategy and ADMM parameter used for single coil experiments were employed.
We generated the square root of sum of squares (SSoS) image from multi-coil reconstructions.

We also validated the performance of ALOHA for accelerated dynamic cardiac data in the $k-t$ domain. 
A cardiac cine data set was acquired using a 3T whole-body MRI scanner (Siemens; Tim Trio) equipped with a 32-element cardiac coil array. The acquisition sequence was bSSFP and prospective cardiac gating was used. The imaging parameters were as follows: FOV$ = 300 \times 300 \mbox{mm}^2$, acquisition matrix size$ = 128\times128$, TE/TR $= 1.37/2.7\mbox{ms}$, receiver bandwidth $= 1184\mbox{Hz/pixel}$, and flip angle $= 40^{\circ}$. The number of cardiac phases was 23 and the temporal resolution was 43.2ms. The k-t space samples  including four lines around zero frequency were retrospectively obtained at the reduction factor of eight according to a Gaussian distribution.
For comparison, k-t FOCUSS \cite{jung2009k}, LORAKS \cite{haldar2014low}, and SAKE \cite{shin2013calibrationless} was used. In case of LORAKS (single coil) and SAKE (multi coil), these algorithms were applied to k-t domain for dynamic reconstruction. The  parameters in k-t FOCUSS, LORAKS, and SAKE, were selected to give the best NMSE values.
 For the  ALOHA in single coil data, the following parameters were used: three level of pyramidal decomposition only along the phase encoding direction,  decreasing LMaFit tolerances ($10^{-1},10^{-2},10^{-3}$) at each scale,  and 17$\times$5 annihilating filter.
 The same  LMaFit rank estimation strategy and ADMM parameter was $\mu=10$. The k-space weighting in Eq. \eqref{eq:wavelet_spline_dyadic} was applied only along the phase encoding direction.

\begin{figure}[!hbt]
\center{
\includegraphics[trim = 0mm 0mm 0mm 0mm,clip=true,width=17cm]{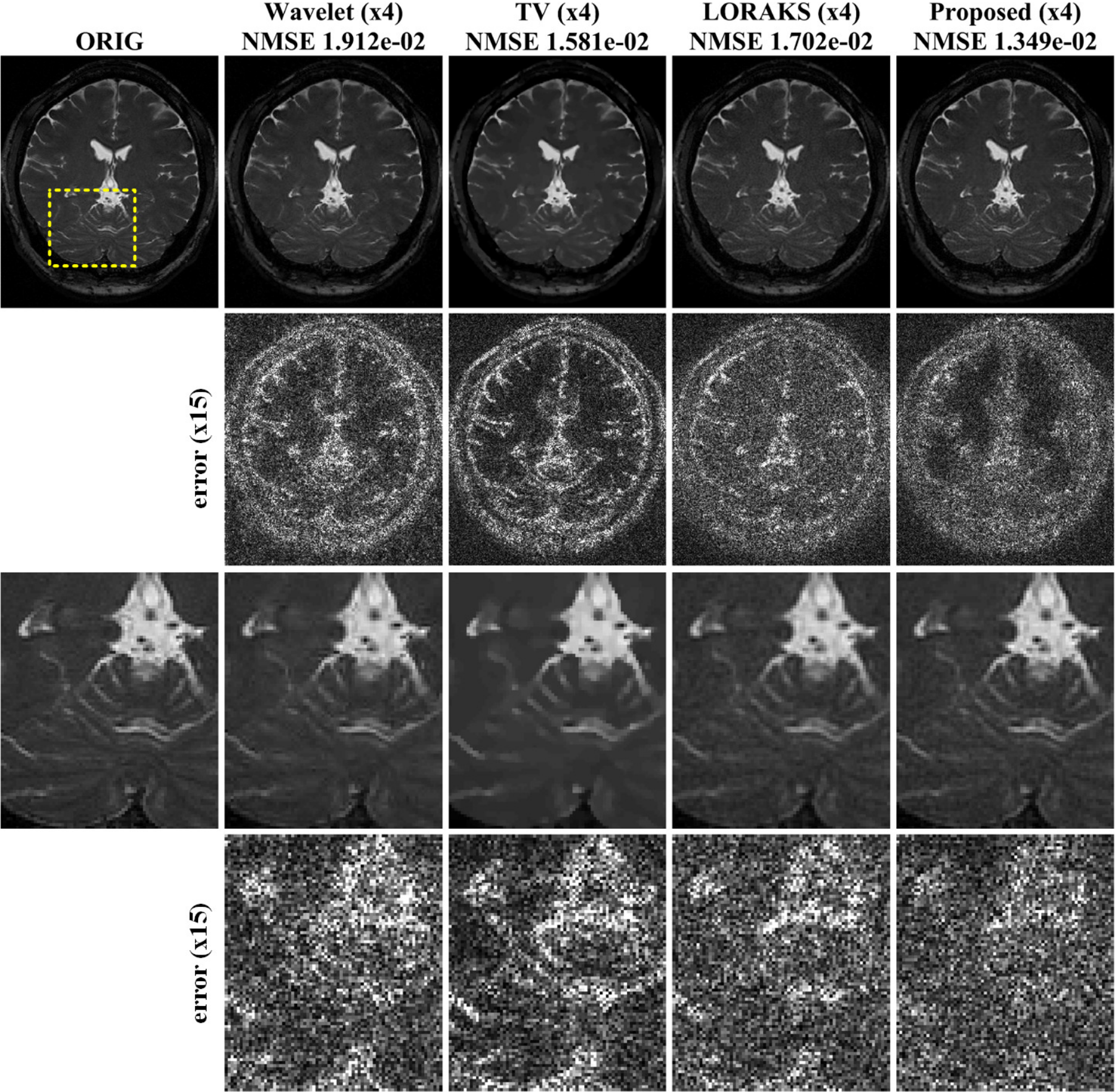}}
\caption{
Comparison with $l_1$-wavelet compressed sensing, TV compressed sensing, LORAKS,  and the proposed method at four fold acceleration factors. The data was acquired from a single channel  coil. The first row shows reconstructed images,  and the second row shows difference images between the ground-truth and the reconstructions, and the third row shows the magnified views of distorted regions in the reconstructed images. The last row shows the difference images in the magnified views.}
\label{fig:single_coil_brain}
\end{figure}

Next, we investigated the synergetic improvement of dynamic imaging from multi-channel acquisition. Four representative  coils out of 32 were used. 
The reason we chose only four coils was to verify that the generalised ALOHA can maximally exploit the multi-channel diversity even with the small number of coils. The four coils were chosen such that it covers every area of images. The same four coils were used for all algorithms for fair comparison.
In the ALOHA, the annihilating filter size was $7\times 7$.
 The same  LMaFit rank estimation strategy and ADMM parameter used before were employed.
After the reconstructions of k-space samples, the inverse Fourier transform was applied, and the SSoS images were obtained by combining the reconstructed images. 

\begin{figure}[!hbt]
\center{
\includegraphics[trim = 0mm 0mm 0mm 0mm,clip=true,width=17cm]{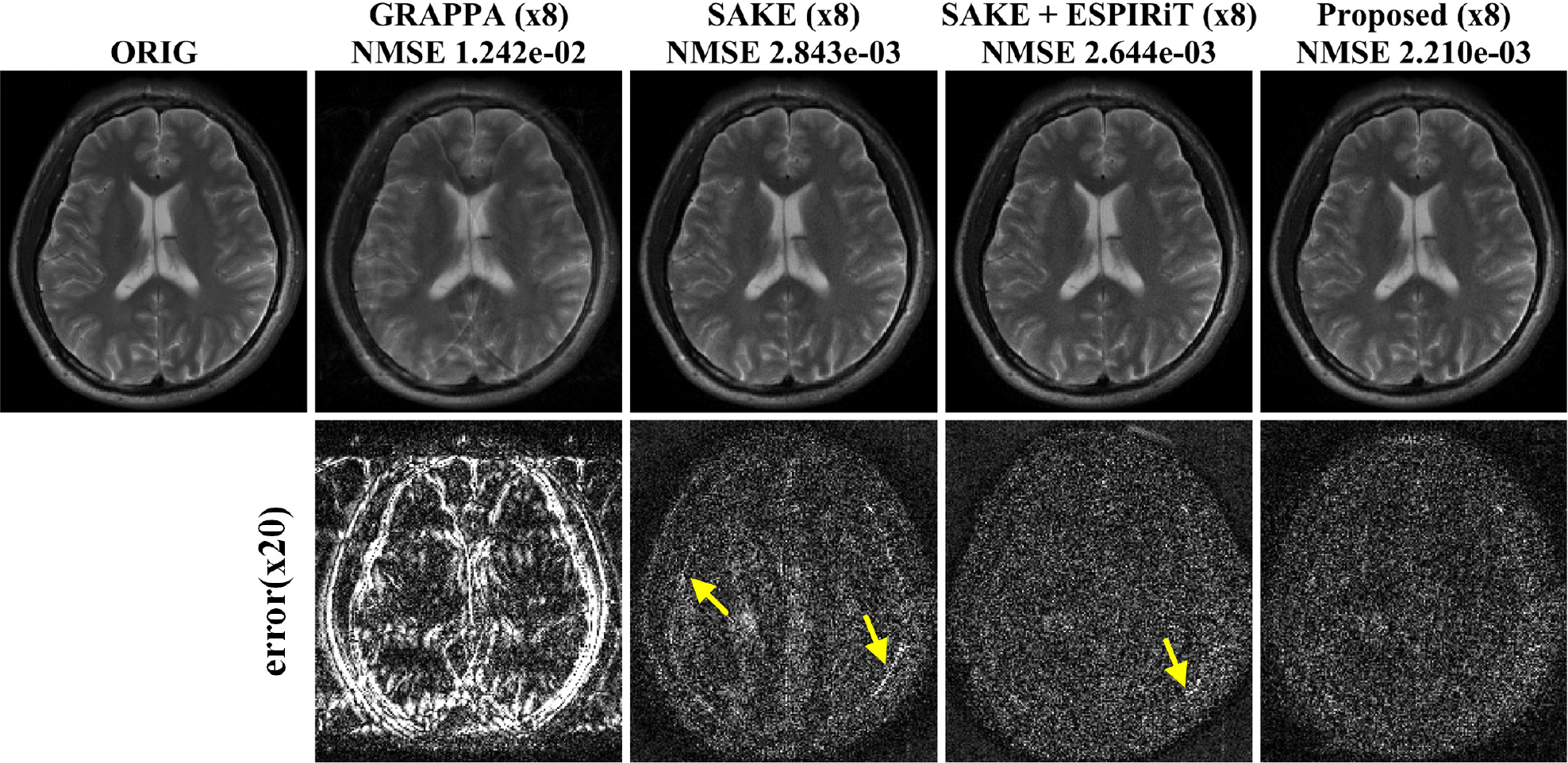}}
\caption{Parallel imaging results using GRAPPA, SAKE, SAKE with ESPIRiT and the proposed method at 8 fold acceleration. The second row shows the difference images. 
Areas with systematic artefacts are indicated by yellow arrows.  Note that  GRAPPA requires ACS lines, so  with the additional 50 samples along ACS, the effective downsampling ratio was 4.785.}
\label{fig:multi_brain}
\end{figure}

\section{Results} 
\label{sec:result}

\subsection{Static MR experiments} 	

Reconstructed results from a single coil brain data are shown in Fig. \ref{fig:single_coil_brain} with the NMSE values.
From the NMSE values, we observed that the performance  ALOHA was quantitatively superior to the performance of $l_1$-wavelet and TV based compressed sensing approach. 
The reconstruction results by ALOHA has less perceivable distortion compared to those of $l_1$-wavelet and TV approaches.  This can be easily observed from the difference images
in the second and the fourth rows of Fig. \ref{fig:single_coil_brain}.
In the case of $l_1$-wavelets and TV, structural distortion around the image edges was easily recognizable. In the last row of Fig. \ref{fig:single_coil_brain}, the edges were reconstructed  accurately by ALOHA.  On the other hand,  the contrast  between grey matters and white matters in $l_1$-wavelets and TV reconstruction were  significantly distorted compared with that of ALOHA. 
The LORAKS reconstruction appeared better than that of  $l_1$-wavelets and TV reconstruction, but there were still remaining errors around the edges and the NMSE value was much 
higher than that of ALOHA.

\begin{figure}[!t]
\center{
\includegraphics[trim = 0mm 0mm 0mm 0mm,clip=true,width=17cm]{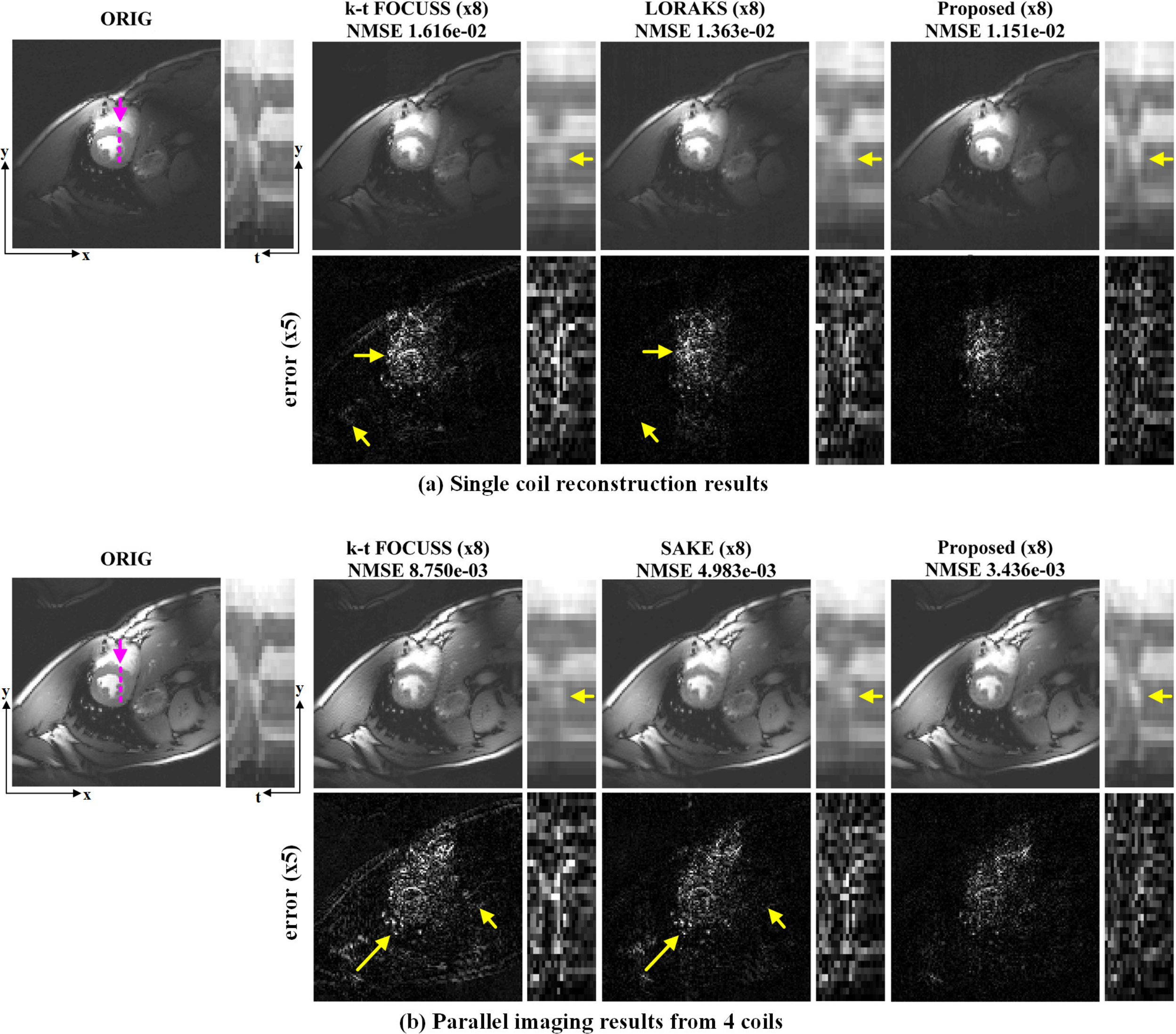}}
\caption{Reconstruction results from 8 fold accelerated k-space data using (a) single coil and (b) four coils data set. Purple lines denote the regions corresponding to y-t cross sections that are magnified along temporal axis. The second rows in both (a) and (b) show the difference images between the ground-truth and the reconstructions.}
\label{fig:dynamic}
\end{figure}

Next we compared our parallel imaging results with those of the existing approaches for additional multichannel brain data set. The NMSE results in Fig. \ref{fig:multi_brain} showed that ALOHA was most accurate. From the difference images at the second row of Fig. \ref{fig:multi_brain}, we observed that proposed method provided reconstruction results more accurately than other algorithms. In SAKE, the structures were distorted around the inner skull and the boundaries of tissues. In SAKE with ESPIRiT, overall reconstruction errors were  higher than those from ALOHA and there were still remaining errors around the skull. 
The reconstruction time was 22.2sec with our preliminary GPU implementation of ALOHA,  which attained a speed-up factor of 5 compared to CPU implementation.
On the other hand, the computational time for MATLAB version of  GRAPPA, SAKE, and SAKE+ESPIRiT  were 9s, 320.4s, and 21.6s, respectively.

\subsection{Dynamic MR experiments} 	

Using the sub-sampled k-space data  at the acceleration factor of eight,
the average NMSE values of k-t FOCUSS, LORAKS, and ALOHA, were $1.616\times 10^{-2}$, $1.363\times 10^{-2}$, and, $1.151\times 10^{-2}$, respectively.
The sub-sampled data was collected  according to a Gaussian distribution and included the four center lines around zero frequency.
 The average NMSE values were calculated using all temporal frames.
These results confirmed that the proposed method outperformed k-t FOCUSS and LORAKS. As shown in Fig.~\ref{fig:dynamic}(a), the temporal profile (indicated as a broken purple line) of the proposed reconstruction provided more accurate 
structures especially in the systole phase that were comparable to the true one, whereas the temporal variation in the k-t FOCUSS and LORAKS reconstruction became smoother and more blurry along the temporal dimension. 

The NMSE values of the parallel dynamic imaging results from k-t FOCUSS, SAKE, and the proposed method using four coil k-space data were $8.75\times 10^{-3}$, $4.983\times 10^{-3}$ and $3.436\times 10^{-3}$, respectively, which quantitatively showed that the proposed method outperformed k-t FOCUSS and SAKE.
Reduced residual artifacts were perceivable  in the ALOHA difference images in Fig.~\ref{fig:dynamic}(b). Moreover, the temporal profiles of the proposed reconstruction  showed more accurate structures  which were comparable to the true one,  whereas the dynamic slice profile  from k-t FOCUSS and SAKE showed smoother and more blurry transition. In particular, the proposed algorithm resulted in more accurate reconstructions of dynamic  changes at the heart wall in the systole phase as shown in Fig.~\ref{fig:dynamic}(b).

\section{Discussion} 	
\label{sec:discussion}

\subsection{Effects of k-space weighting}

Figures~\ref{fig:weighting}(a)(b) illustrate the effects of  wavelet weighting schemes. 
In order to also demonstrate the sensitivity with respect to annihilating filter size,
we calculated the NMSE values by changing the filter size.  Furthermore, to decouple the confounds originated from different implementation of C-based LORAKS and the ALOHA,
the experimental results were generated using the same ALOHA framework with identical parameter setting, except for those related to the weighting.
Figure~\ref{fig:weighting}(a) showed the single channel brain reconstruction results at the acceleration factor of 6.
With weighting, the ALOHA reconstruction was conducted using the following parameters: three levels of pyramidal decomposition, and decreasing LMaFit tolerance values ($5\times10^{-2},5\times10^{-3},5\times10^{-4}$) 
at each level of the pyramid. In addition, an initial rank estimate for LMaFit started with one and was refined automatically in an increasing sequence, and the ADMM parameter was $\mu=10^3$. Note that these parameters were same with those for \eqref{fig:single_coil_brain}.
For the case of non-weighted implementation of ALOHA, the other parameters are exactly the same except the weighting.  
The results showed that the NMSE values of the weighted ALOHA are consistently better than those of unweighted ALOHA regardless of the annihilating filter size.
Moreover, the NMSE values were not sensitive to the annihilating filter size for the cases of the proposed ALOHA, whose NMSE values converged.
The reconstruction results from two implementation at the minimum NMSE values (marked as stars in Figures~\ref{fig:weighting}(a)(b) ) were illustrated, which again clearly showed that the residual errors of the weighted ALOHA is significantly smaller than the unweighted version of ALOHA.

Similar results were obtained from the dynamic cardiac imaging data in Figure~\ref{fig:weighting}(b).
With weighting, the ALOHA reconstruction was conducted using the following parameters: three levels of pyramidal decomposition, and decreasing LMaFit tolerance values ($10^{-1},10^{-2}, 10^{-3}$) at each level of the pyramid. In addition, an initial rank estimate for LMaFit started with one and was refined automatically in an increasing sequence, and the ADMM parameter was $\mu=10$. 
The reconstruction parameters for the unweighted implementation of ALOHA were exactly the same except the weighting.
The results showed that the unweighted ALOHA is very sensitive to the annihilating filter size, which showed the divergent behavior as the filter size increases.
However, the proposed weighted ALOHA exhibited the convergent behaviors. This result clearly confirmed Proposition~\ref{prp:duality} saying that the low rank structure is invariant as long as the annihilating filter size is bigger than the transform domain sparsity level. 
The reconstruction results at the minimum NMSE values clearly showed that the proposed approach provided more clearly transition between diastole and systole phases.

 \begin{figure}[!ht]
\center{
\includegraphics[width=15cm]{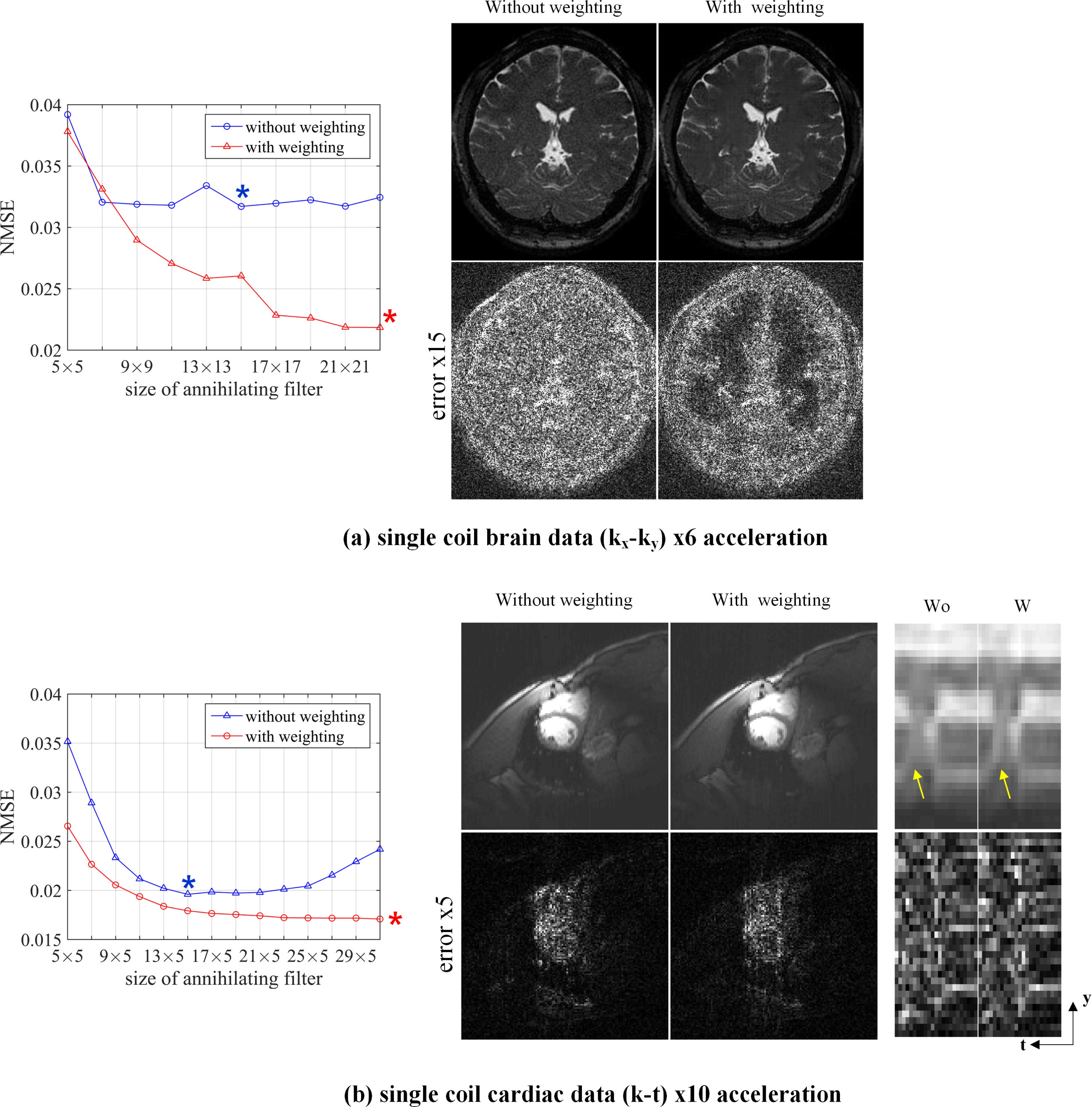}}
\caption{The effect of weighting and annihilating filter size in ALOHA implementation. (a) Single coil brain data results, and (b) single coil dynamic cardiac imaging results.
The results clearly showed that the weighting is necessary for ALOHA to exploit the transform domain sparsity. On the other hand, the unweighted implementation exhibited
divergent behavior with increasing filter size.}
\label{fig:weighting}
\end{figure}

\subsection{Hyper-Parameter Estimation}

\begin{figure}[!bt]
\center{
\includegraphics[height=6cm, width= 9cm]{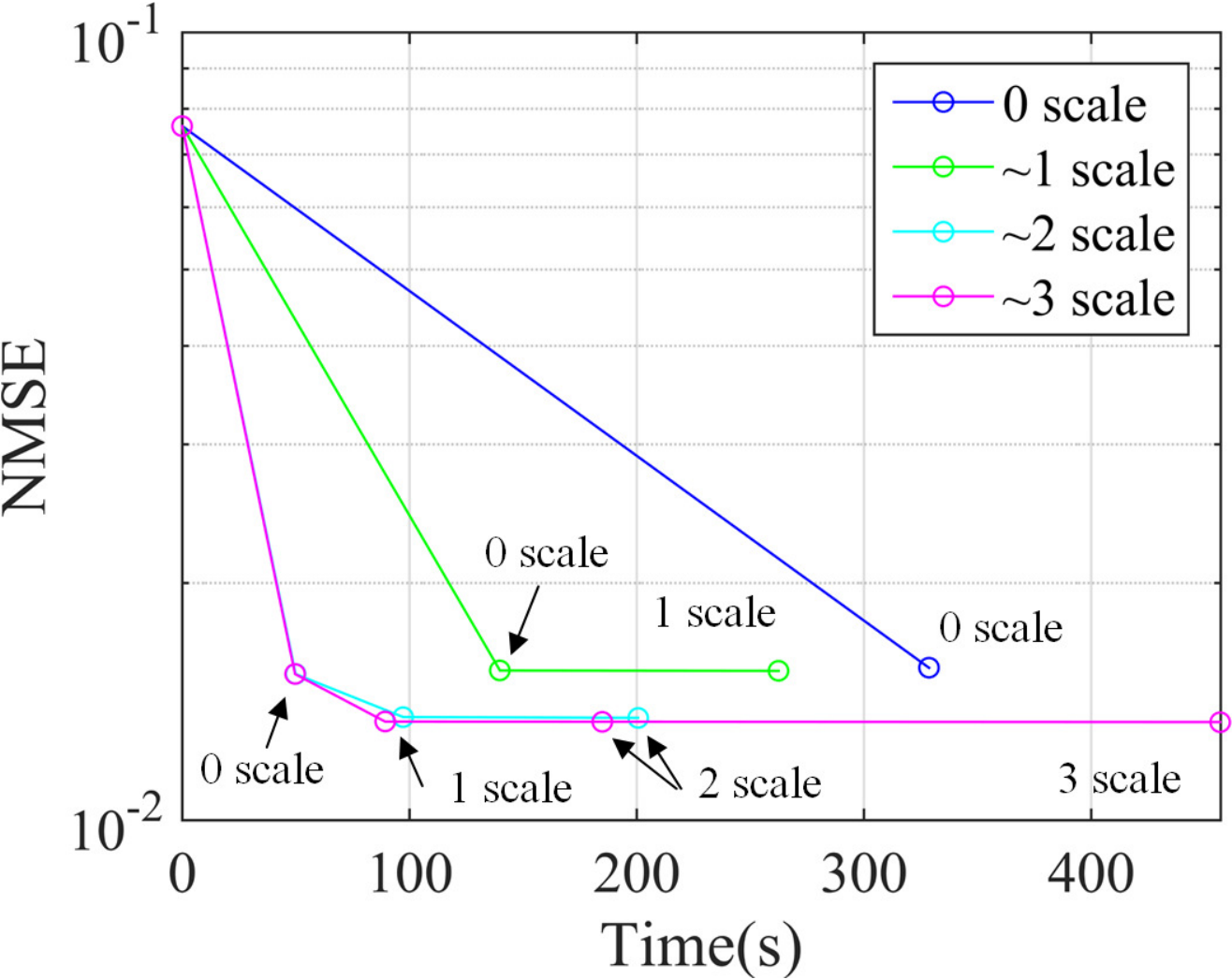}}
\caption{Reconstruction NMSE values with respect to different decomposition levels.  The data is from single channel static MR  reconstruction results  in Fig. \ref{fig:single_coil_brain}. 
}\label{fig:time_per_level}
\end{figure}

In order to show the importance of the pyramidal decomposition,  Fig. \ref{fig:time_per_level} plots the computational time versus NMSE values  of single coil brain data by changing the number of decomposition levels.  
As discussed before,  the maximum decomposition level is determined by the acquired low frequency components and the annihilating filter size.
In this data set, k-space dimension was $208\times 512$ and the annihilating filter size was $23 \times 23$, and $7\times 7$  k-space data around zero frequency were acquired. By converting the Hankel matrix dimension in \eqref{eq:3dhank}  for the $s$-scale, the number of rows should not be smaller than the number of columns; so we should impose constraint
%
$$n_1/2^s-p_1+1 = 208/2^s - 23 +1 \geq  23,$$ 
where $s$ is the scale. This provides $s\leq 2$, and the maximum scale becomes 2. 
{
 Fig. \ref{fig:time_per_level} 
showed that the performance gain increases as a scale increases; as expected,  for $s\geq 3$,  the performance improvement was negligible. 
Therefore,  the maximum scale was determined as $s=2$. 
}

{
The other important reconstruction parameters  include   the size of the annihilating filter, the number of iterations,  and tolerances used in LMaFit algorithm.
Recall that  the annihilating filter size corresponds to the matrix pencil size in sensor array signal processing \cite{hua1990matrix}, and it should be
set larger than the sparsity level of the transform coefficients.    In fact, by considering the expression of $c_s$ in \eqref{eq:cs}, the optimal  annihilating filter size for 2D data
is $n_1/2 \times m_1/2$, where $n_1 \times m_1$ denotes the full k-space data size.
However, such large annihilating filter size introduces significant  computational burden, so we tried to reduce the filter size as long as the image quality is not degraded. 
Based on extensive experiments,  we found that  in single coil image, the filter size  should be set larger than that of parallel imaging,  because the annihilating filter size is solely determined by the sparsity level.
In parallel imaging,  there exist additional annihilating filters from the intercoil relationship, so the annihilating filter size for each Hankel matrix construction can be set  smaller than that of a single coil imaging.
}


Finally, the tolerance level for LMaFit, which corresponds to the fitting accuracy,  plays key role in 
 determining the initial rank estimate.
The initial rank estimate need not be close to the exact rank, but it was used to define the dimension of $U$ and $V$ in ADMM. 
We found that  the tolerance level could be determined by considering different noise contributions in k-space data.
Specifically, higher frequency components are usually contaminated by higher level of noises compared to the low frequency k-space data, so the LMaFit
 fitting accuracy  need not be enforced strictly. This was the case when LMaFit was applied at a lower scale, since high frequency k-space data are more weighted and noises were boosted. 
 On the  other hand, more accurate fitting is required for higher scale data where the lower frequency k-space data are more weighted.
 Consequently, we chose decreasing values of tolerances per scale for in-vivo experiments.


%

\subsection{Concatenation Direction  of Hankel Matrices for Parallel MRI}

\begin{figure}[!bt]
\centering
\includegraphics[trim = 0mm 0mm 0mm 0mm,clip=true,width=4 in]{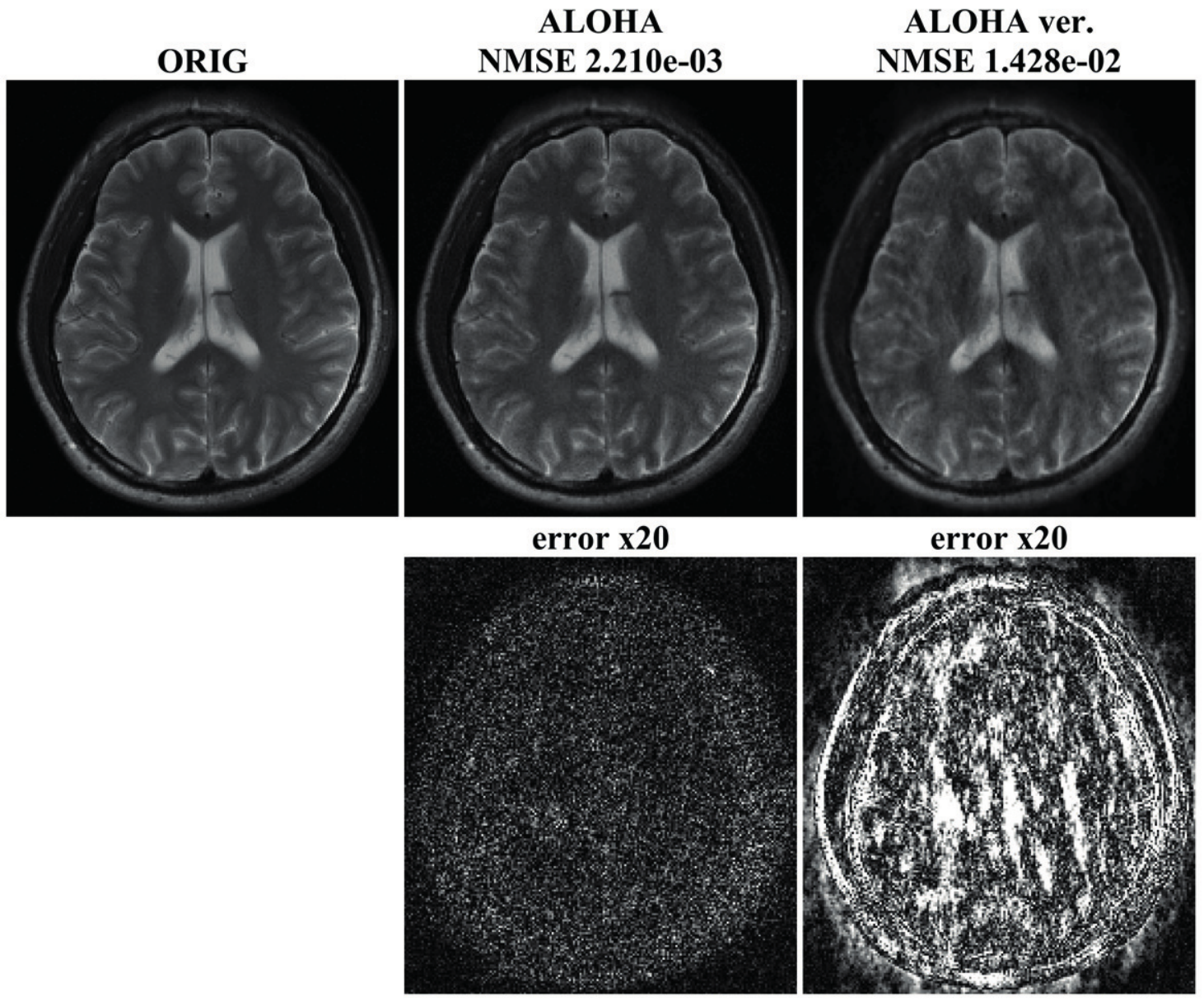}
\caption{Dependency of concatenation direction of Hankel matrix for parallel MRI: (middle) side by side augmentation of Hankel matrices, and (right) vertical augmentation of Hankel matrices. Side-by-side augumentation provided the superior reconstruction results.}
\label{fig:augmentation}
\end{figure}

When generalizing the Hankel matrix model from single to multiple channels, we chose to stack the Hankel matrices corresponding to individual channels side by side. 
 One  could suggest  stacking them one on top of another, which would represent a different matricization. 
Based on our discussion in Section~\ref{sec:theory},  in order to stack the Hankel matrix on top of another, the Hankel matrices should share the same annihilating filter. 
 However,  for the TV signal, the structure in \eqref{eq:lg} showed that it is not possible to find a common annihilating filter $\hat h_c(\omega)$ such that
 $\hat h_c(\omega) \ast \left(\hat l(\omega)\hat g_i(\omega)\right)= 0,  \forall i.$
On the other hand,     the relationship in Eq.~\eqref{eq:an_mat2} clearly showed that
there are inter-coil annihilating filters that can cancel out each combinations of the k-space signals.  Therefore,  to utilize  \eqref{eq:an_mat2},  we should stack the
Hankel matrix size by side.  

To confirm the theoretical findings, we compared the reconstruction results from the two different stacking of Hankel matrices.
As shown in Fig. \ref{fig:augmentation}, the vertical augmentation of Hankel matrices provided inferior reconstruction results compred to the  side by side augmentation of Hankel
matrices.
This again confirmed the importance of the concatenation direction and the analysis by Proposition~\ref{prp:prank}.

\subsection{Further Extensions}

Note that this work is a first step towards unifying sparse and low rank models,  so there are many rooms for improvement.
In particular, it could be extended for other measurement and sparsity models.
For example,  while Eq.~\eqref{eq:EMaC2}  is based on a data equality constraint that
focuses on standard Cartesian image reconstruction,  we could convert
$(P)$ in \eqref{eq:EMaC2}  to deal with the cases of general non-Cartesian imaging and/or noisy measurements:
\begin{eqnarray}\label{eq:EMaC5}
(P')
 &\min_{\mb\in \Cd^{n} } & \rank \hank (\mb)  \\
&\mbox{subject to } & \| P_\Omega(\mb) - A_\Omega(\hat \lb \odot \hat \fb) \| \leq \delta  \nonumber  \  ,
\end{eqnarray}
 where $\mb$ denotes the spectrum data on cartesian grid, 
 $A_\Omega$ denotes a linear mapping that interpolates non-cartesian  data $\hat \lb \odot \hat \fb$ to the nearest cartesian grid index $\Omega$, 
 and
  the noisy level $\delta$ is determined by the gridding or measurement noises.
 The minimization problem $(P')$ can  be also addressed using an SVD-free ADMM approach similar to \eqref{eq:ADMM}, except that the indicator function-based data fidelity term is
 changed to $l_2$ data fidelity term.
 The detailed algorithm is described in \cite{ye2015compressive}.
 
 In addition, ALOHA could be extended for 
 reconstruction models that incorporate various system non-idealities.  
 For example,  in our recent work \cite{Lee2015fMRI},  we showed that EPI  ghost artifacts that originate from off-resonance related inconsistencies between odd and even echoes can be removed
by exploiting
 that the differential k-space data between the even and odd echoes is a Fourier transform of an underlying sparse image. 
Specifically, we can construct a rank-deficient concatenated Hankel structured matrix from even and odd k-space data, whose missing data can be interpolated using ALOHA.
However, the extension of ALOHA for general off-resonance corrected MR reconstruction is still an open problem, which needs a further investigation in the future.

%
%

While this paper assumes that signals can be sparsified using dyadic wavelet ransforms,   in general,  signals can be more easily  sparsified using non-decimated redundant wavelets.
In this case,
the corresponding non-decimated discrete sample at the $s$-scale is given by
$f^{s}[l]:=\left. \langle f, \psi_s(\cdot - x) \rangle\right|_{x=l}$ 
whose spectrum can be represented by
\begin{eqnarray}
\hat f^{s}(e^{i\omega}):= 
 \sum_{l}   f^{s}[l]e^{-i\omega l}  
&=&
 \frac{1}{\sqrt{2^{s}}} \hat \psi^*(2^{s} \omega) \hat f (\omega).
\end{eqnarray}
Because there exists no  aliasing components due to the lack of  downsampling,  the construction of weighted Hankel matrix is much simpler than 
the dyadic wavelet transform. However, one of the potential downsides is
that  the number of  non-zero wavelet transform coefficients  $f^s[l]$ increases up to $2^sk$.
Therefore, more k-space measurements are required to recover k-space data  corresponding to the coarse scale wavelet coefficients.
One potential solution would be  to utilize the recovered high frequency k-space samples as additional measurements for interpolating the coarser level
k-space samples. 
Note that  this is different from dyadic wavelet case which  discards higher frequency k-space samples  in recovering low frequency k-space samples.
However, the efficacy of this proposal needs to be evaluated systematically, which is beyond the scope of this paper.

Other than wavelet or TV representation, modern compressed sensing approach deals with advanced sparsifying transforms such as
patch-based methods,  dictionary learning,  Markov penalties, and so on. The extension of ALOHA for such model would be  important and rewarding, which  requires  more extensive investigation
in subsequent research. 
Recently,  ALOHA was successfully used for MR parameter mapping \cite{Lee2015MRPM}.
Considering recent success of direct parameter mappings \cite{velikina2013accelerating,ma2013magnetic},  the extension of ALOHA for direct parameter mapping would be very interesting, which
deservers further investigations in the future.

\section{Conclusion}\label{sec:conclusion} 	

In this paper, we proposed a general framework for  annihilating filter based low-rank Hankel matrix approach (ALOHA) for static and dynamic MRI inspired by recent calibration-free k-space methods such as SAKE and LORAKS/P-LORAKS.
Because natural images can be much more effectively sparsified in the transform domains,  we 
 generalized the idea to include signals that
can be sparsfied in the transform domains.  Our analysis showed that the transform domain sparsity can be equivalently represented as
low-rank Hankel structured matrix in the weighted k-space domain, whose weighting function is determined solely by the transform, not by the data.
In particular,  when   signals are effectively sparsified in dyadic wavelet transform,
the corresponding low rank Hankel matrix completion problem can be  implemented using  a pyramidal decomposition, which significantly reduces the overall computational
complexity and improves the noise robustness. 
For parallel imaging data, we verified  that  by stacking Hankel matrix from each coil side by side, we  may fully exploit  the coil sensitivity diversity thanks to existence of
inter-coil annihilating filters.

%

Reconstruction results from single coil static MR imaging confirmed that the proposed method outperformed the existing compressed sensing framework with TV regularization. 
We further demonstrated  superior performance of the proposed method  in static parallel MR imaging even without calibration data. Furthermore, the algorithm was successfully extended to dynamic accelerated MRI along k-t domain with both single coil and multi coil dynamic MR data. 
Therefore,  we concluded that the proposed algorithm was very effective in unifying the compressed sensing and parallel MRI.

\section*{Acknowledgement}

The authors would like to thank  Prof. Michael Unser  at EPFL  for the many insightful discussions.
The authors also like to thank Prof. Sung-Hong Park at KAIST for providing the single coil brain MR k-space data set.


%
\end{document}